\documentclass[onecolumn]{IEEEtran}
\IEEEoverridecommandlockouts

\def\astarmsn{a^{*}_{\mbox{\tiny{$m,\!s^{n}$}}}}
\def\Astarmsn{A^{*}_{\mbox{\tiny{$m,\!s^{n}$}}}}

\def\bI{\boldsymbol{I}}

\def\iotapl{\iota_{\mbox{{\tiny$\oplus$}}}}
\def\apl{a_{\mbox{{\tiny$\oplus$}}}}
\def\hatapl{\hata_{\mbox{{\tiny$\oplus$}}}}

\def\hatapl{\hata_{\mbox{{\tiny$\oplus$}}}}

\def\ttp{\mathtt{p}}
\def\ttt{\mathtt{t}}
\def\ttT{\mathtt{T}}

\def\hata{\hat{a}}

\def\hatw{\hat{w}}

\def\hatm{\hat{m}}
\def\hum{\hat{\underline{m}}}

\def\hatv{\hat{v}}

\global\long\def\11{\mathbbm{1}}

\def\ulinetau{\underline{\tau}}
\def\ulines{\underline{s}}
\def\ulineS{\underline{S}}

\def\ulinem{\underline{m}}

\def\ulinee{\underline{e}}
\def\ulineR{\underline{R}}

\def\ulineX{\underline{X}}
\def\ulinex{\underline{x}}

\def\ulineU{\underline{U}}
\def\ulineu{\underline{u}}

\def\ulineCalS{\underline{\mathcal{S}}}
\def\ulineCalX{\underline{\mathcal{X}}}

\def\ulineCalU{\underline{\mathcal{U}}}
\def\ulineCalV{\underline{\mathcal{V}}}

\def\ulineCalM{\underline{\mathcal{M}}}

\def\ulineX{\underline{X}}

\def\ulineV{\underline{V}}
\def\ulinev{\underline{v}}

\def\olinekappa{\overline{\kappa}}

\def\3To1BC{$3-$to$-1$}

\def\define{:{=}~}

\def\naturals{\mathbb{N}}

\def\hatm{\hat{m}}

\def\Expectation{\mathbb{E}}
\def\fieldq{\mathcal{F}_{q}}

\newcommand{\msout}[1]{\text{\sout{\ensuremath{#1}}}}

\newif\ifProofForORDBC

\def\parsec{\par\noindent}
\def\med{\medskip\parsec}

\def\tildem{\tilde{m}}

\def\CalA{\mathcal{A}}

\def\CalE{\mathcal{E}}
\def\CalF{\mathcal{F}}

\def\CalH{\mathcal{H}}

\def\CalL{\mathcal{L}}
\def\CalM{\mathcal{M}}

\def\CalP{\mathcal{P}}

\def\CalS{\mathcal{S}}

\def\CalU{\mathcal{U}}
\def\CalV{\mathcal{V}}
\def\CalW{\mathcal{W}}
\def\CalX{\mathcal{X}}
\def\CalY{\mathcal{Y}}

\def\11{\mathbbm{1}}

\def\3To1BC{$3-$to$-1$}

\def\define{:{=}~}

\def\naturals{\mathbb{N}}

\def\Expectation{\mathbb{E}}

\usepackage{stackengine}

\def\define{\mathrel{\ensurestackMath{\stackon[1pt]{=}{\scriptstyle\Delta}}}}

\newif\ifJournal

\usepackage{amssymb}
\usepackage{amsmath}
\usepackage{mathrsfs}
\usepackage{ulem}
\usepackage{epsf,epsfig}
\usepackage{cite}
\usepackage{color}
\usepackage{dsfont}
\usepackage{bbm}
\usepackage{amsthm}
\usepackage{physics}

\newtheorem{theorem}{Theorem}

\newcommand{\comment}[1]{}

\begin{document}

\sloppy
\newtheorem{remark}{\it Remark}
\newtheorem{thm}{Thm.}
\newtheorem{corollary}{Corollary}
\newtheorem{definition}{Defn.}
\newtheorem{lemma}{Lemma}
\newtheorem{example}{Ex.}
\newtheorem{prop}{Prop.}
\newtheorem{fact}{Fact.}

\title{\huge Communicating over a Classical-Quantum MAC with State Information Distributed at Senders}

\author{
\IEEEauthorblockN{Arun Padakandla\\}
\IEEEauthorblockA{University of Tennessee, USA \\
Email: arunpr@utk.edu
}
}
\maketitle

\begin{abstract}
 We consider the problem of communicating over a classical-quantum (CQ) multiple access channel with classical states non-causally available at the transmitter, henceforth referred to as a QMSTx. QMSTx is a classical-quantum multiple access analogue of the channel studied \cite{1980MMPCIT_GelPin} by Gelfand and Pinsker in 1980. We undertake a Shannon-theoretic study and focus on the problem of characterizing inner bounds to the capacity region of a QMSTx. We propose a new coding scheme based on \textit{union coset codes} - codes possessing algebraic closure properties and derive a new inner bound that subsumes the largest known inner bound based on IID random coding. We identify examples for which the derived inner bound is strictly larger.
\end{abstract}

\section{Introduction}
\label{Sec:Introduction}
Consider the scenario of communicating over a classical-quantum (CQ) multiple access channel with classical states (QMSTx) depicted in Fig.~\ref{Fig:CQMACDSTx}. $S_{1},S_{2}$ model a pair of channel states that are jointly distributed and whose evolution across time is IID. Transmitter (Tx) $j$ is provided the entire realization of the state $S_{j}$ non-causally and is required to communicate its message $M_{j}$ to a receiver (Rx) that is uninformed of the states. If the channel is in state $s_{1},s_{2}$ and Txs $1,2$ choose input symbols $x_{1},x_{2}$ respectively, then the Rx is provided the quantum state $\rho_{x_{1}x_{2}s_{1}s_{2}}$. Our focus is on the problem of characterizing an inner bound to the capacity region of a general QMSTx. In the sequel, we describe our motivation and contributions.

Our motivation in addressing this problem is four fold. Firstly, QMSTx is a network for which the conventional long established technique of proving inner bounds via IID random codes, also referred to herein as unstructured codes, is sub-optimal. In this article, we design a coding scheme based on \textit{union coset codes} (UCC) - codes possessing algebraic closure properties - that strictly outperforms the best known coding scheme based on IID codes. Specifically, we analyze performance of the designed coding scheme to derive an inner bound (Thms.~\ref{Thm:CQMACDSTxFirstStepTheorem}, \ref{Thm:CombineIIDandUCC}) to the capacity region of a QMSTx that, not only subsumes the largest inner bound via unstructured codes, but strictly enlarges the same for identified examples (Ex.~\ref{Ex:BDDQMSTx}, Prop.~\ref{Prop:LinearCodesOutperformIID}). These findings build on our earlier work \cite{201710TIT_PadPra, BkNIT_PraPadShi} and maybe viewed as another step \cite{202107ISIT_AnwPadPra3CQIC, 202107ISIT_AnwPadPraComMAC} in our pursuit of designing coding schemes based on coset codes for network CQ communication.

While the utility of coset codes have been established in several networks \cite{197903TIT_KorMar,200710TIT_NazGas, BkNIT_PraPadShi, 201804TIT_PadPra, 202107ISIT_AnwPadPraComMAC, 202107ISIT_AnwPadPra3CQIC}, their use for a QMSTx is unique and leads us to our second motivation. Coset codes have facilitated higher rates in communication scenarios wherein a compressive bivariate function of the messages or codewords have to be decoded. For instance, on both the $3-$user interference \cite{201603TIT_PadSahPra, 202107ISIT_AnwPadPra3CQIC} and broadcast channels \cite{201804TIT_PadPra}, coset codes enable efficient decoding of the bivariate interference. QMSTx is a CQ MAC wherein both messages need to be decoded and decoding a compressive bivariate function of either the codewords or the messages can lead to obfuscation of the messages. Indeed, coset codes have no role in communication over a CQ-MAC without Tx states. It is therefore natural to question the utility of structured codes in communicating over a QMSTx. A second motivation of our work is to demonstrate how algebraic closure properties can be exploited to \textit{efficiently sieve relevant information} and thereby facilitate enhanced communication over a QMSTx. We illustrate this phenomena in the context of an example (Ex.~\ref{Ex:BDDQMSTx}) and a self contained discussion (Secs.~\ref{SubSec:BDDQMSTx}, \ref{SubSec:RoleOfAlgebraicCodes}). In particular, Sec.~\ref{Sec:IdeaSection} enables us convey the main ideas of this article.

\begin{figure}
    \centering
    \includegraphics[width=4in]{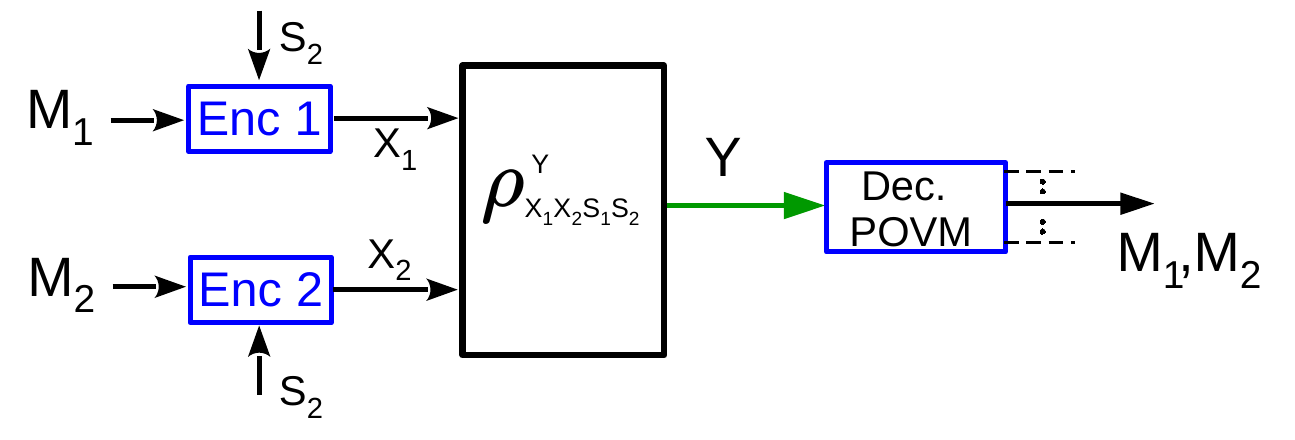}
    \vspace{-0.35in}
    \caption{}
    \label{Fig:CQMACDSTx}
        \vspace{-0.25in}
\end{figure}

Thirdly, this study enables us enrich the family of coset codes for CQ communication beyond nested coset codes (NCC) \cite{202107ISIT_AnwPadPraComMAC, 202107ISIT_AnwPadPra3CQIC} and partitioned coset codes (PCC) \cite{202202arXiv_Pad3CQBC} studied recently. As elaborated in \cite{BkNIT_PraPadShi} and recent works \cite{201706ISIT_HeiShiPra, 201910TIT_HeiShiPra}, coding schemes based on NCC or PCC designed for a classical analogue of a QMSTx, i.e., a classical MAC with states, can be strictly inferior to a coding scheme based on UCC. We have taken this cue and designed UCC based coding schemes and our results in Sec.~\ref{Sec:CQPTPSTxWithUCCs} also prove that UCCs achieve capacity of a single sender CQ channel. Lastly, our findings maybe viewed as developing new strategies to handle various network scenarios arising in an eventual quantum communication network.

Our presentation is organized in a modular fashion. In Sec.~\ref{Sec:IdeaSection}, we illustrate the main ideas of our work through a discussion in the context of a carefully chosen example (Ex.~\ref{Ex:BDDQMSTx}). A general coding scheme for a QMSTx consists of two layers - unstructured codes and UCC. In order to illustrate the new elements in a simplified setting, we present a simplified coding scheme involving only the UCC layer in Sec.~\ref{Sec:Step1InnerBound}, wherein we provide a detailed description and proof steps. In Sec.~\ref{Sec:IIDAndCosetCodes}, we present our inner bound that comprises of both unstructured and UCC layers. In Sec.~\ref{Sec:CQPTPSTxWithUCCs}, we prove that a coding scheme based on UCCs can achieve the best known single-letter inner bound to the capacity region of a single sender version of a QMSTx, i.e., single Tx, single Rx classical-quantum channel with random classical channel states available only at the Tx, referred to therein as a QSTx. In this article, we focus on conveying the ideas and techniques developed. 

The study of channels with Tx state information traces its roots back to Shannon \cite{195804IBMJRD_Sha} and has had considerable influence on other problems. Indeed, Gelfand and Pinkser's coding scheme \cite{1980MMPCIT_GelPin} for the classical single Tx channel with states, henceforth referred to as the Gelfand-Pinsker channel, forms the core of Marton's coding \cite{197905TIT_Mar} for the broadcast channel. Recently, Boche, Cai and N\"otzel \cite{201604JPA_BocCaiNot} studied the CQ analogue of the Gelfand-Pinsker channel and proved achievability of a corresponding inner bound. Their work exploits the method of types and the findings by N\"otzel\cite{201310JPA_Not} in proving achievability of the inner bound. Moreover, their work \cite{201604JPA_BocCaiNot} highlights the difference between the causal and the non-causal availability of state information at the Tx in regards to the single-letterization of the capacity. Our focus is on designing a new coding scheme and characterizing its performance via a single-letter expression. We have not commented on optimality of the bounds derived herein.

\section{Preliminaries and Problem Statement}
\label{Sec:Preliminaries}
For $n\in \mathbb{N}$, $[n] \define \left\{1,\cdots,n \right\}$. $\fieldq$ denotes the finite field of size $q$ and $\oplus_{q}$ denotes addition within $\fieldq$. For a Hilbert space $\mathcal{H}$, $\mathcal{L}(\mathcal{H}),\mathcal{P}(\mathcal{H})$ and $\mathcal{D}(\mathcal{H})$ denote the collection of linear, positive and density operators acting on $\mathcal{H}$ respectively. We let an \underline{underline} denote an appropriate aggregation of pairs of objects. For example, $\ulineCalU \define \CalU_{1}\times \CalU_{2}$ denotes the Cartesian product for sets, $\ulinex\define(x_{1},x_{2}) \in \ulineCalX$ and  $\ulinex^{n}\define(x_{1}^{n},x_{2}^{n})$. The specific aggregation will be clear from context. We abbreviate probability mass function as PMF.

Consider a (generic) \textit{QMSTx} specified through (i) two finite input sets $\CalX_{1},\CalX_{2}$, (ii) two finite sets $\CalS_{1},\CalS_{2}$ of states, (iii) a PMF $\ttp_{\ulineS}(\cdot)$ on $\ulineCalS$, (iii) a collection $(\rho_{\ulinex\ulines} \define \rho_{x_{1}x_{2}s_{1}s_{2}}\in \mathcal{D}(\mathcal{H}_{Y}): (\ulinex,\ulines) \in \ulineCalX\times \ulineCalS )$ of density operators and (iv) cost functions $\kappa_{j} :\CalX_{j} \times \CalS_{j} \rightarrow [0,\infty)$ for $j\in [2]$. The cost function is additive, i.e., having observed the state sequence $s_{j}^{n}$ the cost incurred by sender $j$ in preparing the state $\otimes_{t=1}^{n}\rho_{\ulinex_{t}\ulines_{t}}$ is $\olinekappa_{j}(x_{j}^{n},s_{j}^{n}) \define \frac{1}{n}\sum_{t=1}^{n}\kappa_{j}(x_{jt},s_{jt})$. Reliable communication on a QMStx entails identifying a code. 
\begin{definition}
\label{Defn:CQMACDSTxCode}
An $(n,\ulineCalM,\ulinee,\lambda)$ QMSTx code consists of two message index sets $\CalM_{j} : j \in [2]$, two encoder maps $e_{j}:[\CalM_{j}] \times \CalS_{j}^{n} \rightarrow \CalX_{j}^{n}$ and a decoder POVM $\lambda \define \{ \lambda_{\ulinem}=\lambda_{m_{1},m_{2}} \in \CalP(\CalH^{\otimes n}) : \ulinem \in \ulineCalM\}$. The average error probability of the code is
\begin{eqnarray}
 \label{Eqn:AvgErrorProb}
 \overline{\xi}(\ulinee,\lambda) \define 1-\frac{1}{|\ulineCalM|}\sum_{\ulinem \in \ulineCalM}\sum_{\ulines^{n} \in \ulineCalS^{n}}\ttp^{n}_{\ulineS}(\ulines^{n})\tr(\lambda_{\ulinem}\rho_{\ulinem,\ulines^{n}}).
 \nonumber
\end{eqnarray}
where $\rho_{\ulinem,\ulines^{n}} \define \otimes _{t=1}^{n}\rho_{\ulinex_{t}\ulines_{t}}$ and $(x_{j1},\cdots,x_{jn}) =  e_{j}(m_{j},s_{j}^{n})$. Average cost incurred by sender $j$ in transmitting $m_{j}$ is $\tau_{j}(e_{j}|m_{j}) \define \sum_{s_{j}^{n}}\ttp_{S_{j}}^{n}(s_{j}^{n})\kappa_{j}(e_{j}(m_{j},s_{j}^{n}),s_{j}^{n})$ and the average cost incurred by sender $j$ is $\tau_{j}(e_{j}) \define \frac{1}{|\CalM_{j}|}\sum_{m_{j}}\tau_{j}(e_{j}|m_{j})$.
\end{definition}
The object of interest is the capacity region of a QMSTx defined below. In this article, we focus on characterizing inner bounds to the capacity region of a QMSTx.
\begin{definition}A rate-cost quadruple $(\ulineR,\ulinetau) \in [0,\infty)^{4}$ is \textit{achievable} if there exists a sequence of QMSTx codes $(n,\ulineCalM^{(n)},\ulinee^{(n)},\lambda^{(n)})$ for which $\displaystyle\lim_{n \rightarrow \infty}\overline{\xi}(\ulinee^{(n)},\lambda^{(n)}) = 0$,
\begin{eqnarray}
 \label{Eqn:3CQICAchievability}
 \lim_{n \rightarrow \infty} n^{-1}\log \mathcal{M}_{j}^{(n)} = R_{j}, \mbox{ and }\lim_{n \rightarrow \infty} \tau_{j}(e_{j}^{(n)}) \leq \tau_{j} .
 \nonumber
\end{eqnarray}
The capacity region $\mathscr{C}$ of the QMSTx is the set of all achievable rate-cost vectors and $\mathscr{C}(\ulinetau) \define \{ \ulineR:(\ulineR,\ulinetau) \in \mathscr{C}\}$.
\end{definition}

\section{Role of Algebraic Structure/Closure}
\label{Sec:IdeaSection}
In this section, we explain \textit{how} and \textit{why} structured codes can facilitate enhanced communication over a QMSTx. We begin by reviewing the best known unstructured coding scheme.
\subsection{Joint Decoding of Unstructured Codes}
\label{SubSec:UnstructuredCodingScheme}
A QMSTx is a `MAC extension' of a single sender CQ channel with random states \cite{201604JPA_BocCaiNot}. A coding scheme for the QMSTx can therefore be obtained by combining the Gelfand-Pinsker encoding scheme \cite{1980MMPCIT_GelPin} with a simultaneous decoder of a MAC \cite[Thm.~2]{201206TIT_FawHaySavSenWil}. Specifically, each sender $j$ builds a $U_{j}-$code (Tab.~\ref{Fig:CodingForCQMACDSTx}) on an auxillary set $\CalU_{j}$. The $U_{j}-$code comprises of $2^{n(R_{j}+B_{j})}$ codewords partitioned into $2^{nR_{j}}$ bins. The message $m_{j} \in [2^{nR_{j}}]$ indexes a bin and the encoder looks for a codeword within this bin that is jointly typical with the state sequence $s_{j}^{n}$. The chosen codeword, denoted as $u_{j}^{n}(m_{j},s_{j}^{n})$, and the state sequence $s_{j}^{n}$ are mapped to an input sequence in $\CalX_{j}^{n}$.

The decoder POVM performs simultaneous decoding on the $U_{1},U_{2}-$codebooks. Specifically, one can adopt the decoding POVM proposed in proof of \cite[Thm.~2]{201206TIT_FawHaySavSenWil}. Leveraging the `projector trick' therein and the `overcounting trick' \cite{201512TIT_SavWil} in the context of Marton decoding, we can analyze the error probability and derive an inner bound. The latter is the largest known inner bound achievable via any unstructured coding scheme and we provide a characterization of the same below.
\begin{theorem}
 \label{Thm:UnstructuredCodingTheorem}A rate-cost quadruple $(\ulineR,\ulinetau)$ is achievable if there exists finite sets $\CalU_{1},\CalU_{2}$, conditional distributions $p_{X_{j},U_{j}|S_{j}}$ on $\CalX_{j} \times \CalU_{j}$ for $j \in [2]$ such that $p_{\ulineS\ulineU \ulineX}(\ulines,\ulineu,\ulinex) = \ttp_{\ulineS}(\ulines)\prod_{j=1}^{2}p_{X_{j}U_{j}|S_{j}}(x_{j},u_{j}|s_{j})$ with respect to which
 \begin{eqnarray}
 \label{Eqn:UnstructuredCodeRateRegion1}
 R_{j} \!<\! I(U_{j};Y,U_{\msout{j}})_{\sigma}\!-\!I(U_{j};S_{j})_{\sigma}, \Expectation\{\!\kappa_{j}\!(\!X_{j},S_{j}\!)\} \!\leq\! \tau_{j},\nonumber\\
 R_{1}+R_{2} < I(\ulineU;Y)_{\sigma}+I(U_{1},S_{1};U_{2},S_{2})_{\sigma},\
  \nonumber
 \end{eqnarray}
for $j \in [2]$, where all informations are computed wrt the state
\begin{eqnarray}
 \label{Eqn:QuantStateInIIDCodeRateRegion}
 \sigma^{Y\ulineX\ulineS\ulineU} \define\! \sum_{\ulines,\ulinex,\ulineu}\!\!p_{\ulineS\ulineU\ulineX} (\ulines,\ulineu,\ulinex)\rho_{\ulinex\ulines} \otimes \ketbra{\ulinex~\ulineu~\ulines}.
 \end{eqnarray}
\end{theorem}
\begin{table}\centering
\begin{minipage}{.3\textwidth}
 \centering
\begin{tabular}{|c|c|c|}
\hline\hline
$w$&$t$&$\gamma(w,t)$\\
\hline
0&0&$\ketbra{0}$\\
\hline
0&1&$\ketbra{v_{\theta}^{\perp}}$\\
\hline
1&0&$\ketbra{1}$\\
\hline
1&1&$\ketbra{v_{\theta}}$\\
\hline\hline
\end{tabular}\vspace{0.05in}\caption{$\rho_{x_{1}x_{2}s_{1}s_{2}}\!\!=\!\gamma(x_{1}\!\oplus\! x_{2},s_{1}\!\oplus\! s_{2})$}
\label{Table:Ex1Density1Operator}
\end{minipage}\begin{minipage}{2.95in}\centering
\includegraphics[width=2.45in]{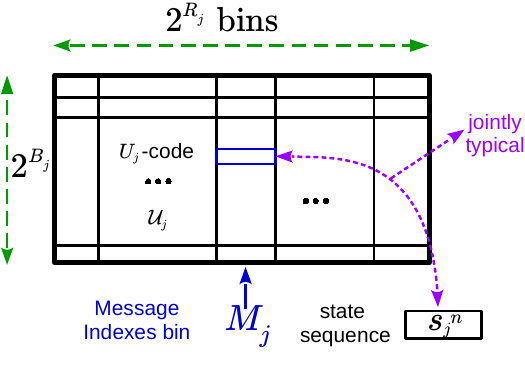}
\caption{Encoding rule for sender $j$.}
\label{Fig:CodingForCQMACDSTx}
\end{minipage}
\end{table}

\subsection{Binary Double Dirty MAC}
\label{SubSec:BDDQMSTx}
Our discussion for the following example portrays the deficiency of unstructured codes and the role of structure.
\begin{example}
 \label{Ex:BDDQMSTx}
 Let $\CalX_{1}=\CalX_{2}=\CalS_{1}=\CalS_{2}=\{0,1\}$, $\ttp_{\ulineS}(\ulines)=\frac{1}{4}$ for every $\ulines \in \ulineCalS$, $\ket{v_{\theta}} = [\cos\theta ~ \sin\theta]^{T}$ and $\ket{v_{\theta}^{\perp}} = [\sin\theta~-\!\cos\theta]^{T}$. For $(\ulinex,\ulines) \in \{0,1\}^{4}$, let $\rho_{x_{1}x_{2}s_{1}s_{2}} = \gamma(x_{1}\!\oplus \!x_{2},s_{1}\!\oplus\!s_{2})$, where $\gamma(\cdot,\cdot)$ is provided in Table \ref{Table:Ex1Density1Operator}, $\oplus$ denotes addition in the binary field $\CalF_{2}$ and the cost function $\kappa_{j}(x_{j},s_{j})=\mathds{1}_{\{x_{j}=1\}}$ is the Hamming weight function. For a $\tau \in (0,\frac{1}{2})$, what is $\mathscr{C}(\tau,\tau)$?
\end{example}
Our discussion below for the $\theta=0$ case leads us pedagogically to the non-commuting case $\theta \in (0,\frac{\pi}{2})$ which follows. The classical channel corresponding to $\theta=0$ was first studied by Philosof and Zamir \cite{200906TIT_PhiZam} and the following discussion describes their findings.

\med\textit{Case $\theta=\!0\!$} : Since the collection $\left(\rho_{\ulinex\ulines}: (\ulinex,\ulines) \in \{0,1\}^{4}\right)$ is commuting, we identify this as a classical MAC with distributed states whose output $Y\!\in \!\{0,1\}$, inputs $X_{1},X_{2}\!\in\! \{0,1\}$ and states $S_{1},S_{2}\!\in\! \{0,1\}$ are related as $Y=X_{1}\oplus S_{1}\oplus X_{2} \oplus S_{2}$. $S_{1},S_{2}$ are uniformly distributed and the average Hamming weight of the inputs is constrained to $\tau < \frac{1}{2}$. The latter constraint precludes the senders from negating the effect of the state. What rate pairs are then achievable?

We first study the best unstructured coding scheme and characterize the corresponding largest known inner bound. Towards that end, observe that the effective classical channel of Ex.~\ref{Ex:BDDQMSTx} is a `MAC extension' of a single sender channel with random parameters whose output $Y \in \{0,1\}$, Hamming cost-constrained input $X \in \{0,1\}$ and uniformly distributed state $S \in \{0,1\}$ are related as $Y=X\oplus S$. Philosof and Zamir \cite{200906TIT_PhiZam} proved that the best unstructured coding strategy for Ex.~\ref{Ex:BDDQMSTx} is obtained by replicating, at both the senders, the capacity achieving strategy for the single sender channel. Specifically, they prove the optimal choice of parameters in Thm.~\ref{Thm:UnstructuredCodingTheorem} for Ex.~\ref{Ex:BDDQMSTx} to be binary auxillary sets $\CalU_{1}=\CalU_{2}=\{0,1\}$, $p_{U_{j}|S_{j}}(1|0)=p_{U_{j}|S_{j}}(0|1)=\tau=1-p_{U_{j}|S_{j}}(0|0)=1-p_{U_{j}|S_{j}}(1|1)$ and $X_{j}=U_{j}\oplus S_{j}$ for $j\in [2]$.

We now detail the coding scheme corresponding to the above choice to shed light on its deficiency. To communicate at rate $R_{j} < h_{b}(\tau)$, sender $j$ randomly partitions the entire set of $2^{n}$ sequences into $2^{nR_{j}}$ bins. The message $m_{j}$ indexes the bin within which the sender looks for a codeword that is within an average Hamming distance of $\tau$ from the observed state sequence. Since each bin contains $2^{n(1-R_{j})} > 2^{n(1-h_{b}(\tau))}$ sequences chosen randomly, the sender finds such a codeword whp. Let $U_{j}^{n}$ denote the chosen codeword and $S_{j}^{n}$ the observed state sequence. Sender $j$ inputs $X_{j}^{n} = U_{j}^{n} \oplus S_{j}^{n}$ on the channel. The choice of the $U_{j}-$codeword guarantees that the Hamming weight constraint is met.

Observe that the channel relationship $Y=X_{1}\oplus S_{1}\oplus X_{2} \oplus S_{2}$ implies that the received vector is $Y^{n}=U_{1}^{n}\oplus U_{2}^{n}$. Recall that each message $m_{j}$ of sender $j$ is assigned a bin of $U_{j}-$codewords. The space of received sequences occupied by a \textit{single message pair} $(m_{1},m_{2})$ is therefore got by adding all possible codeword pairs in the two bins indexed by $m_{1},m_{2}$. Since the codewords in each bin is picked uniformly and independently without any joint structure, every pair yields whp a distinct sum, resulting in the range of this addition to be of size $2^{n(2-R_{1}-R_{2})} > 2^{2n(1-h_{b}(\tau))}$. Since the `fan-out' of every message pair is of size atleast $2^{2n(1-h_{b}(\tau))}$, we cannot hope to pack more than $\frac{2^{n}}{2^{2n(1-h_{b}(\tau))}}$ fan-outs in the binary output space resulting in the following fact.

\begin{fact}
 \label{Fact:UnstructRatesForBDDMAC}
 Consider Ex.~\ref{Ex:BDDQMSTx} with average Hamming cost constraint $\tau< \frac{1}{2}$. Any rate pair $(R_{1},R_{2})$ achievable by unstructured coding schemes satisfies $R_{1}+R_{2}<uce\{\max\{0,2h_{b}(\tau)-1\}\}$ where $uce\{f(\tau)\}$ denotes the upper convex envelope of the function $f(\tau)$. See \cite{200906TIT_PhiZam} for a proof.
\end{fact}
We now present a linear coding scheme that can achieve any rate pair $(R_{1},R_{2})$ satisfying $R_{1}+R_{2} < h_{b}(\tau)$. For simplicity, we describe achievability of the rate pair $(h_{b}(\tau),0)$. Our coding scheme is identical to the unstructured coding scheme with two key differences. The first key difference is that the bins of each sender's codebook are chosen to be \textit{cosets of a common linear code}. Let $\lambda_{2}$ denote a linear code of rate $1-h_{b}(\tau)$ whose cosets can quantize a uniform source to with an average Hamming distortion of $\tau$. In other words, a uniformly and randomly chosen coset of $\lambda_{2}$ contains a codeword within an average Hamming distance of $\tau$ of the observed state sequence whp. See \cite{197401TIT_Wyn} or \cite{BkNIT_PraPadShi} for proof of existence. Since sender $2$ has no message to transmit, it is provided with just $\lambda_{2}$ that serves as its only bin. Sender $1$ is provided with all of the $2^{nh_{b}(\tau)}$ cosets of $\lambda_{2}$, each of which serves as its bins. The encoding is identical to that for unstructured coding. 

The codebook of sender $2$ when added to any bin of user $1$'s code results in a coset of
$\lambda_{2}$, and therefore contains at most $2^{n(1-h_{b}(\tau))}$
codewords. Moreover, since $U_{2}^{n}$ lies in $\lambda_{2}$, user $1$'s codeword
$U_{1}^{n}$ and the received vector $Y^{n}=U_{1}^{n} \oplus_{2} U_{2}^{n}$ lie in the
same coset. The receiver can identify the coset from which sender $1$ chose his $U_{1}-$codeword and hence gather sender $1$'s message. Since the channel is
noiseless, sender $1$ may employ all cosets of $\lambda_{2}$ and therefore communicate at
rate $h_{b}(\tau)$ which is larger than $2h_{b}(\tau)-1$ for all $\tau \in
(0,\frac{1}{2})$.

\med\textit{Case $\theta \in (0,\frac{\pi}{2})$} : The arguments in \cite{200906TIT_PhiZam} can be used to prove that the optimal choice of parameters in Thm.~\ref{Thm:UnstructuredCodingTheorem} for this case too is $\CalU_{1}=\CalU_{2}=\{0,1\}$, $p_{U_{j}|S_{j}}(1|0)=p_{U_{j}|S_{j}}(0|1)=\tau = 1-p_{U_{j}|S_{j}}(0|0)=1-p_{U_{j}|S_{j}}(1|1)$ and $X_{j}=U_{j}\oplus S_{j}$ where $\oplus$ denotes addition mod$-2$. This implies the quantum state corresponding to which we compute our information quantities is
\begin{eqnarray}
 \sigma^{YS_{1}S_{2}X_{1}X_{2}U_{1}U_{2}} \!=\! \sum_{s_{1},s_{2}}\! \frac{\tau(1\!-\!\tau)}{4}\!\left[ \mathds{1}_{\left\{\substack{s_{1}\oplus s_{2}\\=0}\right\}}\ketbra{1}+\mathds{1}_{\left\{\substack{s_{1}\oplus s_{2}\\=1}\right\}}\ketbra{v_{\theta}} \right]\otimes\ketbra{s_{1}~s_{2}}\otimes \left[\!\!\begin{array}{c}\ketbra{0~1~s_{1}~1\oplus s_{2}}+\\\ketbra{1~0~1\oplus s_{1}~s_{2}} \end{array}\!\!\right]
 \nonumber\\
 + \sum_{s_{1},s_{2}}\left[ \mathds{1}_{\left\{\substack{s_{1}\oplus s_{2}\\=0}\right\}}\ketbra{0}+\mathds{1}_{\left\{\substack{s_{1}\oplus s_{2}\\=1}\right\}}\ketbra{v_{\theta}^{\perp}} \right]\otimes \ketbra{s_{1}~s_{2}}\otimes\left[\!\!\begin{array}{c}\frac{(1-\tau)^{2}}{4}\ketbra{0~0~s_{1}~s_{2}}+\\\frac{\tau^{2}}{4}\ketbra{1~1~1\oplus s_{1}~1\oplus s_{2}} \end{array}\!\!\right].
 \nonumber
\end{eqnarray}
The bound on the sum rate achievable using IID random codes as stated in Thm.~\ref{Thm:UnstructuredCodingTheorem} is $I(U_{1}U_{2};Y)_{\sigma}-I(U_{1};S_{1})-I(U_{2};S_{2})_{\sigma}$. In Appendix \ref{AppSec:BDDQMSTxExampleQuantumStates}, we have provided characterization of the component quantum states with respect to which the above information quantities have to be computed. Referring to the same, it can be verified that $I(U_{1}U_{2};Y)_{\sigma}-I(U_{1};S_{1})-I(U_{2};S_{2})_{\sigma} = \alpha-2+2h_{b}(\tau)$ where
\begin{eqnarray}
\label{Eqn:RateForNonZeroTheta}
 \alpha = \tilde{h}_{b}((1-2\tau)^{2}\sin\theta)-\tilde{h}_{b}(\sqrt{1-4\epsilon(1-\epsilon)\sin^{2}\theta}), ~ \tilde{h}_{b}(x)\define h_{b}\left(\frac{1}{2}+\frac{x}{2}\right)\mbox{ and }\epsilon=2\tau(1-\tau).
 \end{eqnarray}
It maybe verified that $\alpha=1$ if $\theta=0$ indicating the maximum sum rate achievable is a continuous function of $\theta$ as one expects. In Prop.~\ref{Prop:LinearCodesOutperformIID}, we verify that the linear coding scheme achieves any rate pair satisfying $R_{1}+R_{2} < uce\{\max\{0,\alpha-1+h_{b}(\tau)\}\}$ which strictly subsumes that achievable above.

Two important observations are in order. Firstly, since exponentially many pairs of codewords from $\lambda_{2}$ and the coset chosen by sender $1$ have the same sum, the receiver cannot disambiguate the pair of codewords chosen by the two senders. It can only disambiguate the sum $U_{1}^{n}\oplus U_{2}^{n}$. The second key difference in this coding scheme is that the receiver must not attempt to decode the pair, but instead decipher the message by decoding the sum of the two codewords.

\subsection{Sieving Relevant Information via Algebraic Closure}
\label{SubSec:RoleOfAlgebraicCodes}
The key difference between the structured and unstructured coding scheme is the decoding rule. While the former pins down the pair, the latter only decodes the sum, leaving uncertainity in the pair. Note that, the codeword $u_{j}^{n}(m_{j},s_{j}^{n})$ chosen by sender $j$ contains, in addition to the message, information about $s_{j}^{n}$. By requiring the receiver to pin down the pair $(u_{j}^{n}(m_{j},s_{j}^{n}):j \in [2])$ of chosen codewods, the unstructured coding scheme is forcing the receiver to gather information of the state seqeuences that is not of value to it. Is there a function of $(u_{j}^{n}(m_{j},s_{j}^{n}):j \in [2])$ that, while containing information of the pair $m_{1},m_{2}$ of messages can also suppress the amount of information of the pair $s_{1}^{n},s_{2}^{n}$ and can the coding scheme enable the Rx decode this function efficiently? The structured coding scheme is enabling the Rx do this via the mod$-2$ function.

\section{Inner Bound based on Union Coset Codes}
\label{Sec:Step1InnerBound}

\begin{theorem}
 \label{Thm:CQMACDSTxFirstStepTheorem}A rate-cost quadruple $(\ulineR,\ulinetau)$ is achievable if there exists a finite field $\CalV_{1}=\CalV_{2}=\CalW=\fieldq$ and conditional PMFs  $p_{X_{j}V_{j}|S_{j}}$ on $\CalX_{j}\times\CalV_{j}$ for $j \in [2]$ with respect to which
 \begin{eqnarray}
  R_{1}+R_{2} \!<\! \min\{H(V_{1}|S_{1})_{\sigma}:j \in [2]\}\!-\! H(V_{1}\oplus_{q}V_{2}|Y )_{\sigma}\!\!\!
 \end{eqnarray}
where all mutual informations are computed wrt the state
 \begin{eqnarray}
  \sigma^{Y\!\ulineX\ulineV\ulineS} \define \!\!\!\!\!\sum_{\ulines,\ulinev,w,\ulinex}\!\!\!\!p_{\ulineS\ulineV W \ulineX}(\ulines,\ulinev,w,\ulinex)\rho_{\ulinex\ulines} \!\otimes\! \ketbra{\ulinex~\ulinev~w~\ulines}\mbox{with}\nonumber\\
  p_{\ulineS\ulineV W \ulineX}(\ulines,\ulinev,w,\ulinex)\! = \ttp_{\ulineS}(\ulines)\!\prod_{j=1}^{2}\!p_{X_{j}V_{j}|S_{j}}(x_{j},v_{j}|s_{j})\mathds{1}_{\left\{\substack{w=\\v_{1}\oplus_{q}v_{2}}\right\}}.\nonumber
 \end{eqnarray}
 for all $(\ulines,\ulinev,w,\ulinex)\in\ulineCalS\times\ulineCalV\times\CalW\times\ulineCalX$.
 \end{theorem}
\begin{proof}
We begin by outlining our techniques and identifying the new elements. The main novelty is in the code structure we design and the decoding POVM we propose. In Sec.~\ref{SubSec:Step1CodeStructure}, we characterize a UCC and describe our codes. The Gelfand-Pinsker encoding (Sec.~\ref{SubSec:Step1Encoding}) is employed by both senders. We decode only the sum codeword and hence employ a single user decoding POVM (Sec.~\ref{SubSec:Step1DecodingPOVM}). Since we decode into a UCC obtained by adding two statistically correlated UCCs, our analysis is not a standard one and detailed in Sec.~\ref{SubSec:Step1ProbErrorAnalysis}.

\subsection{Code Structure}
\label{SubSec:Step1CodeStructure}
The gain in rates for Ex.~\ref{Ex:BDDQMSTx} crucially relied on the bins of both codes being coset shifts of a common linear code, thereby ensuring that the size of the sum of any pair of bins was contained. We observe that the shifts can be arbitrary and there are no structural requirement on the union of these cosets. We are thus led to a UCC.
\begin{definition}
\label{Defn:UCC}
A UCC built over $\fieldq$ is specified through a generator matrix $g \in \fieldq^{k \times n}$ and a map $\iota:\fieldq^{l} \rightarrow \fieldq^{n}$ of coset shifts. The collection $c(m)\define \{v^{n}(a,m)=ag\oplus_{q} \iota(m) \}$ forms the bin corresponding to message $m \in \fieldq^{l}$ and the union $\displaystyle \cup_{u}c(m)$ of bins forms the code. We refer to this as an $(n,k,l,g,\iota)$ UCC or an $(n,k,l,g,c)$ UCC.
\end{definition}
We employ UCCs as the codebook for both senders. The symmetry in Ex.~\ref{Ex:BDDQMSTx} permitted us to design codes of the same rate for both senders. In general, to enable codes of different rates, we propose a `nesting' of the two UCCs. Without loss of generality, assume the size of sender $1$'s bins is the smaller of the two. We equip user $j$ with UCC $(n,k_{j},l_{j},g_{j},\iota_{j})$ and enforce $g_{2} = \left[ g_{1}^{T}~g_{2/1}^{T} \right]^{T}$. This ensures that the bins of user $1$'s code are sub-cosets of the bins of user $2$'s code, thus guaranteeing the desirable property mentioned prior to Defn.~\ref{Defn:UCC}. Let $\lambda_{j} \define (v_{j}^{n}(a_{j},m_{j}) \define a_{j}g_{j}\oplus_{q} \iota_{j}(m_{j}) : (a_{j},m_{j}) \in \CalV^{k_{j}} \times \CalV^{l_{j}}) $ denote the codebook of sender $j$

\subsection{Encoding}
\label{SubSec:Step1Encoding}
Our encoding is identical to that described for unstructured codes in Sec.~\ref{SubSec:UnstructuredCodingScheme}. On observing message $m_{j} \in [q^{l_{j}}]$ and state sequence $s_{j}^{n}$, sender $j$ looks for a codeword in $c_{j}(m_{j})$ that is jointly typical with $s_{j}^{n}$. It it finds atleast one, one among these is chosen and denoted $v_{j}^{n}(m_{j},s_{j}^{n})$. If it finds none, $v_{j}^{n}(m_{j},s_{j}^{n})$ is set to a default codeword in $c_{j}(m_{j})$. The pair $(s_{j}^{n},v_{j}^{n}(m_{j},s_{j}^{n}))$ is mapped to an input sequence via a `fusion map' $f_{j}:\CalS_{j}^{n} \times \CalV_{j}^{n} \rightarrow \CalX_{j}^{n}$. For the sake of the ensuing analysis, we formalize this encoding with some notation.

Let $\alpha_{j}(m_{j},s_{j}^{n}) \define \sum_{a_{j}}\mathds{1}_{\{(v_{j}^{n}(a_{j},m_{j}),s_{j}^{n}) \in T_{\eta}(p_{S_{j}V_{j}}) \}}$ be the number of available jointly typical codewords and let
\begin{eqnarray}
 \label{Eqn:Step1ProofList}
 \CalL_{j}(m_{j},s_{j}^{n}) \define
 \begin{cases}
  \{ \substack{a_{j} : (v_{j}^{n}(a_{j},m_{j}),s_{j}^{n}) \in T_{\eta}(p_{V_{j}S_{j}}) }\}&\mbox{if }\substack{\alpha_{j}(m_{j},s_{j}^{n}) \geq 1}\\
  \{0^{k_{j}}\}&\mbox{ otherwise}
 \end{cases}
 \nonumber
\end{eqnarray}
For every pair $(m_{j},s_{j}^{n})$, $a_{j}(m_{j},s_{j}^{n})$ is an element chosen from $\CalL_{j}(m_{j},s_{j}^{n})$. We define $v_{j}^{n}(m_{j},s_{j}^{n}) \define v_{j}^{n}(a_{j}(m_{j},s_{j}^{n}),s_{j}^{n})$. A predefined `fusion map' $f_{j}:\CalS_{j}^{n} \times \CalV_{j}^{n} \rightarrow \CalX_{j}^{n}$ is used to map the pair $s_{j}^{n},v_{j}^{n}(m_{j},s_{j}^{n})$ to an input sequence in $\CalX_{j}^{n}$ henceforth denoted $x_{j}^{n}(m_{j},s_{j}^{n})$.

\subsection{Decoding POVM}
\label{SubSec:Step1DecodingPOVM}
Consider the UCC $(n,k_{2},l_{1}+l_{2},g_{2},\iotapl)$ where $\iotapl(\ulinem)  = \iota_{1}(m_{1})\oplus_{q}\iota(m_{2})$ and let $w^{n}(a,\ulinem)  \define ag_{2}\oplus_{q}\iota_{1}(m_{1}) \oplus_{q}\iota_{2}(m_{2})$ denote its codewords. Let $\pi_{a,\ulinem}$ be the conditional typical projector of $\otimes_{t=1}^{n}\rho_{w_{t}(a,\ulinem)}$ wrt the state $\rho_{w} = \sum_{\ulinex,\ulines}p_{\ulineX\ulineS|W}(\ulinex,\ulines|w)\rho_{\ulinex\ulines} : w \in \CalW$ where $p_{\ulineS\ulineX W}$ is as defined in the Thm. statement. We define $\gamma_{a,\ulinem} \define \pi^{Y}\pi_{a,\ulinem}\pi^{Y}$ where $\pi^{Y}$ is the unconditional typical projector of the state $\sum_{\ulinex,\ulines}p_{\ulineX\ulineS}(\ulinex,\ulines)\rho_{\ulinex\ulines}$. The decoding POVM is
\begin{eqnarray}
 \label{Eqn:DecodingPOVM}
 \lambda_{\ulinem} \define\left( \sum_{\hata,\hatm_{1},\hatm_{2}}\!\!\!\!\!\!\gamma_{\hata,\hatm_{1},\hatm_{2}}\right)^{-\frac{1}{2}}\!\!\!\!\!\!\sum_{a}\gamma_{a,\ulinem}\left(\sum_{\hata,\hatm_{1},\hatm_{2}}\!\!\!\!\!\!\gamma_{\hata,\hatm_{1},\hatm_{2}} \right)^{-\frac{1}{2}}
\end{eqnarray}
and $\lambda_{-1} \define I^{\otimes n} - \sum_{\ulinem}\lambda_{\ulinem}$.

\subsection{Probability of Error Analysis}
\label{SubSec:Step1ProbErrorAnalysis}
We begin our analysis by stating the distribution of the random code. Specifying $g_{1},g_{2/1},\iota_{j}(m_{j}),a_{j}(m_{j},s_{j}^{n}), x_{j}^{n}(m_{j},s_{j}^{n}) : (m_{j},s_{j}^{n}) \in [q^{l_{j}}]\times \CalS_{j}^{n}$ completely specifies the code. A distribution for the random code is therefore specified through a distribution of these objects. We let upper case letters denote the corresponding random objects. $(G_{2},G_{2/1},\iota_{j}(m_{j}):m_{j}\in [q^{l_{j}}]: j \in [2]$ are mutually independent and uniformly distributed on their respective range spaces. $A_{j}(m_{j},s_{j}^{n})$ is conditionally independent of the earlier mentioned objects given $\alpha_{j}(m_{j},S_{j}^{n})$ and uniformly distributed in $\CalL(m_{j},s_{j}^{n})$. Finally, given all of the earlier mentioned objects, $X_{j}^{n}(m_{j},s_{j}^{n}) \sim p_{X|VS}^{n}(\cdot|V_{j}(m_{j},S_{j}^{n}),S_{j}^{n})$.

The average probability of error is
\begin{eqnarray}
 \label{Eqn:ProofStep1ErrAnalysis-1}
 \lefteqn{\xi(\ulinee,\lambda) =\! \sum_{\ulinem}\!\frac{\zeta(\ulinem)}{|\ulineCalM|}\!\mbox{ where } \zeta(\ulinem\!) \define \sum_{\ulines^{n}}\ttp_{\ulineS}^{n}(\ulines^{n})\zeta(\ulinem|\ulines^{n})}\!\!\!\!
 \\
 &&\!\!\!\!\!\!\!\!\!\!\!\!\!\zeta(\ulinem|s^{n}\!)\! \define\! \tr\{\!(\boldsymbol{I}\! -\! \lambda_{\ulinem}\!)\rho_{\ulinem,\ulines^{n}} \!\}, \rho_{\ulinem,\ulines^{n}} \!\!\displaystyle\define \bigotimes_{t=1}^{n}\!\rho_{x_{1}\!(m_{1},s_{1}^{n}\!)_{t}x_{2}\!(m_{2},s_{2}^{n}\!)_{t}\ulines_{t}}
 \nonumber
\end{eqnarray}
where $\boldsymbol{I}=I^{\otimes n}$, $\CalM_{j} = [q^{l_{j}}]$ and hence $|\ulineCalM| = q^{l_{1}+l_{2}}$. Henceforth, we focus on a generic term $\zeta(\ulinem)$. Let $\apl \define a_{1}(m_{1},s_{1}^{n})~0^{k_{2}-k_{1}} \oplus a_{2}(m_{2},s_{2}^{n})$ and $\CalE_{j} \define \{ \alpha_{j}(m_{j},s_{j}^{n}) \geq 1\}$, $\CalE \define \CalE_{1}\cap\CalE_{2}$. With these, we have,
\begin{eqnarray}
\label{Eqn:ProofStep1ErrAnalysis-11}
\lefteqn{\zeta(\ulinem|\ulines^{n}) \leq \ttt_{0}+\ttt_{1}+\ttt_{2}\mbox{ where }\ttt_{0} \define \mathds{1}_{\CalE_{1}^{c}}+\mathds{1}_{\CalE_{2}^{c}}}\\
 &&\!\!\!\!\!\!\!\!\!\!\!\ttt_{1}=\tr\{(\bI-\gamma_{\apl,\ulinem})\rho_{\ulinem,\ulines^{n}}\}\mathds{1}_{\CalE},~\! \ttt_{2} \!\define\!\!\!\! \sum_{\substack{\hatapl \neq \apl\\\hum\neq\ulinem}}\!\!\!\!\tr(\gamma_{\hatapl,\hum}\rho_{\ulinem,\ulines^{n}})\mathds{1}_{\CalE}
 \nonumber
\end{eqnarray}
where, we recall $\gamma_{a,\ulinem} = \pi^{Y}\pi_{a,\ulinem}\pi^{Y}$ and $\pi_{a,\ulinem}$ is the conditional typical projector of $\otimes_{t=1}^{n}\rho_{w_{t}(a,\ulinem)}$. The rest of our proof derives upper bounds on $\ttT_{i}\define \sum_{\ulines^{n}}\ttp_{\ulineS}^{n}(\ulines^{n})\Expectation\{\ttt_{i}\} $ for $i\in [3]$.
\med\textit{Upper bound on $\ttT_{0}$} : $\CalE_{1},\CalE_{2}$ are events involving only classical probabilities and we refer to \cite[Lemma 7 in Appendix B]{201710TIT_PadPra} for a proof of the following.
\begin{prop}
 \label{Prop:Step1ProofEncodingErrorEvent}
 For any $\eta > 0$,  $\exists N_{\eta} \in \naturals$ such that $\forall n \geq N(\eta)$, $\Expectation\{ T_{0}\} \leq \exp\{ -n\eta\}$ if $\frac{k_{j}\log q}{n} > \log q - H(V_{j}|S_{j})$ for $j \in [2]$.
\end{prop}
To comprehend the above bound, note that codewords of a random UCC are uniformly distributed. The expected number of codewords jointly typical with a typical state sequence $s_{j}^{n}$ is $|T_{\eta}(V_{j}|s_{j}^{n})|q^{k-n}$ whose exponent is $ k\log q-n\log q +n H(V_{j}|S_{j})$. Prop.~\ref{Prop:Step1ProofEncodingErrorEvent} guarantees the latter exponent is positive.
\med\textit{Upper bound on $\ttT_{1}$} : Since $\ttt_{1}$ involves the event $\CalE = \CalE_{1}\cap\CalE_{2}$, an upper bound on $\ttT_{1}$ can be derived using pinching and gentle operator lemma. Since this is fairly straightforward we refer the reader to a subsequent version of this manuscript for details.

\med\textit{Upper bound on $\ttT_{2}$} : In our study, we have assumed $k_{2} \geq k_{1}$, i.e., the size of bins in sender $2$'s code to be larger of the two. Under this assumption, we get only one bound on $\frac{k_{1}+2k_{2}+l_{1}+l_{2}}{n}\log q$, but in general, we get two bounds that are stated below.
\begin{prop}
 \label{Prop:Step1ProoT1}
 For any $\eta > 0$, there exists $N_{\eta} \in \naturals$ such that for all $n \geq N(\eta)$, $\ttT_{2} \leq \exp\{ -n\eta\}$ if 
 \begin{eqnarray}
 \max_{j=1,2}\left\{\!\frac{k_{j}}{n}\!\right\}\!+\!\sum_{i=1}^{2}\!\frac{k_{i}\!+\!l_{i}}{n}\! < \!3 -\frac{H(W|Y)_{\sigma}\!-\!\sum_{i=1}^{2}\!H(V_{i}|S_{i})_{\sigma}}{\log q}
 \nonumber
 \end{eqnarray}
\end{prop}
\begin{proof}
Proof is provided in Appendix \ref{AppSec:BoundOnT2}.
\end{proof}
By eliminating $\frac{k_{1}\log q}{n},\frac{k_{2}\log q}{n}$ from the bounds in Prop.~\ref{Prop:Step1ProofEncodingErrorEvent} and \ref{Prop:Step1ProoT1} and equating $R_{j}= \frac{l_{j}\log q}{n}$, we obtain the condition stated in the theorem statement.\end{proof}
\begin{prop}
\label{Prop:LinearCodesOutperformIID}
Consider Ex.\ref{Ex:BDDQMSTx} for $\tau \in (0,\frac{1}{2})$ and $\theta =\frac{\pi}{8}$. The inner bound achievable via UCCs is strictly larger than that achievable via unstructured codes.
\end{prop}
\begin{proof}
By choosing $\CalV_{1}=\CalV_{2}=\CalF_{2}$ the binary field and $p_{X_{j}V_{j}|S_{j}}(1,1\oplus s_{j}|s_{j})=\tau=1-p_{X_{j}V_{j}|S_{j}}(0,s_{j}|s_{j})$ for $s_{j}\in \{0,1\}$ and $j\in [2]$ and evaluating the inner bound in Thm.~\ref{Thm:CQMACDSTxFirstStepTheorem}, it can be verified that any rate pair $(R_{1},R_{2})$ satisfying $R_{1}+R_{2} < uce\{\max\{0,\alpha-1+h_{b}(\tau)\}\}$ is achievable where $\alpha$ is as defined in \eqref{Eqn:RateForNonZeroTheta}. See Fig.~\ref{Fig:PlotsForExample} for plots of the rate regions $R_{1}+R_{2} < uce\{\max\{0,\alpha-2+2h_{b}(\tau)\}\}$ and $R_{1}+R_{2} < uce\{\max\{0,\alpha-1+h_{b}(\tau)\}\}$ achievable via IID and structured codes respectively to verify the latter is strictly larger.
\end{proof}

\begin{figure}
    \centering
    \includegraphics[width=6in]{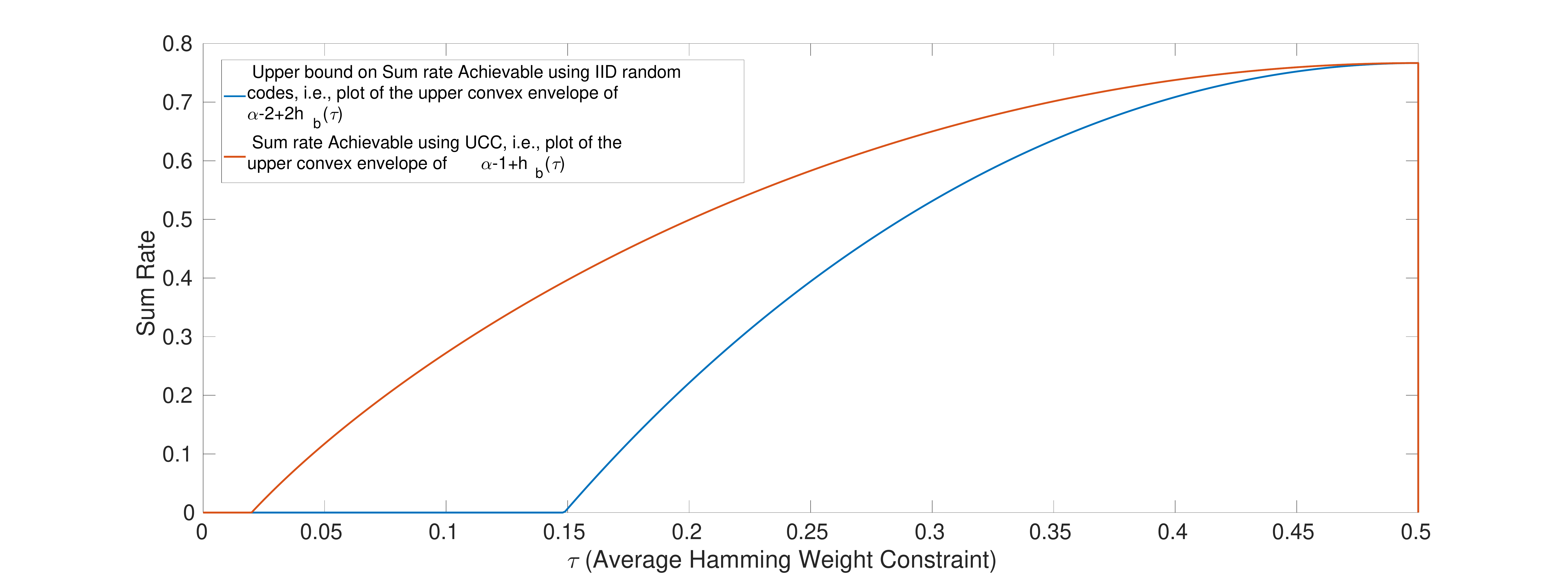}
        \caption{Bound $uce\{\max\{0,\alpha-2+2h_{b}(\tau)\}\}$ on the sum rate achievable via IID random codes is plotted in blue and the sum rate $uce\{\max\{0,\alpha-1+h_{b}(\tau)\}\}$ achievable via UCC is plotted in red.}
    \label{Fig:PlotsForExample}
        \vspace{-0.25in}
\end{figure}

\section{Enhancing IID Coding Schemes via UCCs}
\label{Sec:IIDAndCosetCodes}
The UCC based coding scheme can enable efficient decoding of $V_{1}\oplus V_{2}$. On a QMSTx wherin the latter function contains the information of the pair of messages, the UCC coding scheme can outperform the use of unstructured codes. In general, the information corresponding to the message pair can be embedded in both univariate and bivariate functions of auxillary RVs. A general coding scheme for QMSTx must therefore incorporate both unstructured codes and UCCs. We present the following inner bound that subsumes the inner bounds stated in Thms.~\ref{Thm:UnstructuredCodingTheorem}, \ref{Thm:CQMACDSTxFirstStepTheorem}.

\begin{theorem}
 \label{Thm:CombineIIDandUCC}
 A rate-cost $(\ulineR,\ulinetau)$ quadruple is achievable if there exists finite sets $\CalU_{1},\CalU_{2}$, a finite field $\CalV_{1}=\CalV_{2}=\CalW=\fieldq$ of size $q$ and conditional PMFs $p_{U_{j}V_{j}X_{j}|S_{j}} : j \in [2]$ wrt which
\begin{eqnarray}
\label{Eqn:CombineIIDandUCC1}
\begin{array}{l}R_{j} \leq I(U_{j};U_{\msout{j}}Y)_{\sigma}-I(U_{j};S_{j})_{\sigma}-H(W|\ulineU Y)_{\sigma}\\ ~~~~~~~+\min \left\{H(V_{1}|U_{1}S_{1})_{\sigma},H(V_{2}|U_{2}S_{2})_{\sigma}\right\}, \mbox{ for } j=1,2,\\
R_{1}+R_{2} \leq
I(\ulineU;Y)_{\sigma}-I(\ulineU;\ulineS)_{\sigma}-H(W|\ulineU Y)_{\sigma}\\~~~~~~~~~~~~+
\min \left\{
H(V_{1}|U_{1}S_{1})_{\sigma},H(V_{2}|U_{2}S_{2})_{\sigma}\right\},
\end{array}\nonumber
\end{eqnarray}
where the above entropies are evaluated wrt the state
\begin{eqnarray}
\label{Eqn:CombineIIDandUCC2}
\sigma^{Y\!\ulineX\ulineU\ulineV W\ulineS} \define \displaystyle\!\!\!\!\!\!\sum_{\ulines,\ulineu,\ulinev,w,\ulinex}\!\!\!\!\!\!p_{\ulineS\ulineU\ulineV W \ulineX}(\ulines,\ulineu,\ulinev,w,\ulinex)\rho_{\ulinex\ulines} \!\otimes\! \ketbra{\ulinex~\ulineu~\ulinev~w~\ulines},\nonumber\\
p_{\ulineS\ulineU\ulineV W \ulineX}\left(\substack{\ulines,\ulineu,\ulinev, w,\ulinex}\right)\! = \ttp_{\ulineS}(\ulines)\!\displaystyle\prod_{j=1}^{2}\!p_{X_{j}V_{j}U_{j}|S_{j}}(x_{j},v_{j},u_{j}|s_{j})\mathds{1}_{\left\{\substack{w=\\v_{1}\oplus_{q}v_{2}}\right\}}.\nonumber
 \end{eqnarray}
 for all $(\ulines,\ulinev,w,\ulinex)\in\ulineCalS\times\ulineCalV\times\CalW\times\ulineCalX$.
\end{theorem}
By choosing $\CalV_{1}=\CalV_{2}=\phi$, we can recover the inner bound achievable via IID codes in \ref{Thm:UnstructuredCodingTheorem}. By choosing $\CalU_{1}=\CalU_{2}=\phi$, we can recover the inner bound in Thm.~\ref{Thm:CQMACDSTxFirstStepTheorem}, thus proving that above inner bound subsumes all known inner bounds for a general QMSTx. We now outline the main elements of our proof and furnish details in a subsequent version of this manuscript. The code structure and the encoding is identical to the classical MAC with distributed states and is provided in \cite{201710TIT_PadPra}. Decoding is based on a combination of joint and successive decoding. A joint decoding POVM is employed to decode into $U_{1},U_{2}$. Following this, decoding of $V_{1}\oplus V_{2}$ is performed on the collapsed state. We leverage an alternate form of the `overcounting trick' that we have used in \cite{202202arXiv_Pad3CQBC} to obtain the same pre-Fourier Motzkin bounds as those in \cite[Proof of Thm.~5]{201710TIT_PadPra}.

\section{Communicating over Classical-Quantum Channel with Random States using UCCs}
\label{Sec:CQPTPSTxWithUCCs}
We begin with a formal description of a point-to-point classical quantum channel with classical random states available non-causally at the transmitter. We henceforth refer to this channel as a QSTx.

Consider a (generic) \textit{QSTx} specified through (i) a finite input set $\CalX$, (ii) a finite set $\CalS$ of states, (iii) a PMF $\ttp_{S}(\cdot)$ on $\CalS$, (iii) a collection $(\rho_{xs} \in \mathcal{D}(\mathcal{H}_{Y}): (x,s) \in \CalX\times \CalS )$ of density operators and (iv) cost function $\kappa :\CalX \times \CalS \rightarrow [0,\infty)$. The cost function is additive, i.e., having observed the state sequence $s^{n}$ the cost incurred by the sender in preparing the state $\otimes_{t=1}^{n}\rho_{x_{t}s_{t}}$ is $\olinekappa(x^{n},s^{n}) \define \frac{1}{n}\sum_{t=1}^{n}\kappa(x_{t},s_{t})$. Reliable communication on a QSTx entails identifying a code. 
\begin{definition}
\label{Defn:CQPTPSTxCode}
An $(n,\CalM,e,\lambda)$ QSTx code consists of a message index set $\CalM$, an encoder map $e:\CalM \times \CalS^{n} \rightarrow \CalX^{n}$ with codewords denoted $(x^{n}(m,s^{n})= (x^{n}(m,s^{n})_{t} : 1\leq t \leq n) :(m,s^{n}) \in \CalM\times \CalS^{n} )$ and a decoder POVM $\lambda \define \{ \lambda_{m} \in \CalP(\CalH^{\otimes n}_{Y}) : m \in \CalM\}$. The average error probability of the code is
\begin{eqnarray}
 \label{Eqn:CQPTPSTxAvgErrorProb}
 \overline{\xi}(e,\lambda) \define 1-\frac{1}{|\CalM|}\sum_{m \in \CalM}\sum_{s^{n} \in \CalS^{n}}\ttp^{n}_{S}(s^{n})\tr(\lambda_{m}\rho_{x^{n}(m,s^{n}),s^{n}})\mbox{ where } \rho_{x^{n}(m,s^{n}),s^{n}} = \bigotimes_{t=1}^{n}\rho_{x(m,s^{n})_{t},s_{t}}.
 \nonumber
\end{eqnarray}
Average cost incurred by the sender in transmitting message $m$ is $\tau(e|m) \define \sum_{s^{n}}\ttp_{S}^{n}(s^{n})\kappa(e(m,s^{n}),s^{n})$ and the average cost incurred by the sender is $\tau(e) \define \frac{1}{|\CalM|}\sum_{m}\tau(e|m)$.
\end{definition}
The object of interest is the capacity region of a QSTx defined below. In this section, we prove achievability of the current known largest single-letter inner bounds to the capacity region of a QSTx.

\begin{definition}
\label{Defn:QSTxAchievabilityAndCapacity}
A rate-cost quadruple $(R,\tau) \in [0,\infty)^{2}$ is \textit{achievable} if there exists a sequence of QSTx codes $(n,\CalM^{(n)},e^{(n)},\lambda^{(n)})$ for which $\displaystyle\lim_{n \rightarrow \infty}\overline{\xi}(e^{(n)},\lambda^{(n)}) = 0$,
\begin{eqnarray}
 \label{Eqn:CQPTPSTxCodeAchievability}
 \lim_{n \rightarrow \infty} n^{-1}\log \mathcal{M}^{(n)} = R, \mbox{ and }\lim_{n \rightarrow \infty} \tau(e^{(n)}) \leq \tau.
 \nonumber
\end{eqnarray}
The capacity region $\mathscr{C}$ of the QSTx is the set of all achievable rate-cost vectors and $\mathscr{C}(\tau) \define \{ R:(R,\tau) \in \mathscr{C}\}$.
\end{definition}
\begin{theorem}
 \label{Thm:QSTxSingleLetterInnerBound}
 Consider a QSTx characterized through a finite set $\CalS$ of states, a PMF $\ttp_{S}$ on $\CalS$ modeling the distribution of the random state, an input set $\CalX$ and a collection of density operators $(\rho_{xs} \in \mathcal{D}(\mathcal{H}_{Y}): (x,s) \in \CalX\times \CalS )$. For $\tau > 0$, $R \in \mathscr{C}(\tau)$ if there exists a PMF $\ttp_{S}p_{VX|S}$ on $\CalS\times\CalV\times\CalX$ for which $\sum_{x,s}\ttp_{S}(s)p_{X|S}(x|s)\kappa(x,s) \leq \tau$ and  $R < I(V;Y)-I(V;S)$ where all information quantities are computed with respect to the quantum state
 \begin{eqnarray}
  \label{Eqn:QSTxThmStatementQuantumState}
  \sigma_{YSXV} \define \sum_{x,s,v}\ttp_{S}(s)p_{VX|S}(v,x|s)\rho_{xs}\otimes\ketbra{s~x~v}.
 \end{eqnarray}
\end{theorem}

\begin{proof}
 The two new elements in our proof are the code structure (Sec.~\ref{SubSec:CodeStructureQSTx}). Specifically, we build a union coset code to communicate over the QSTx. Since the codewords of a random union coset code are not mutually independent and are uniformly distributed, a standard information theoretic proof is not applicable. We therefore provide detailed steps in Sec.~\ref{SubSec:QSTXProofErrorProb}, Sec.~\ref{SubSec:QSTxProofCollatingBounds}.
\subsection{Code Structure}
\label{SubSec:CodeStructureQSTx}
Let $\CalV = \fieldq$ be a finite field of size $q$. Consider a $(n,k,l,g,\iota)$ UCC whose codewords are $(v^{n}(a,m)\define ag \oplus_{q} \iota(m) : (a,m) \in \CalV_{k}\times \CalV^{l})$. The message index set $\CalM = [q^{l}]$ and the bin corresponding to message $m$ is the collection $c(m)\define (ag\oplus_{q}\iota(m):a \in \CalV_{k})$. As we describe in the sequel, the encoder observes the state sequence $s^{n}\in \CalS^{n}$ and chooses a codeword in $c(m)$ where $m \in \CalM$ is the message that needs to communicated to the Rx.

\subsection{Encoding}
\label{SubSec:EncodingQSTx}
For every possible pair $(m,s^{n})$ of message and state sequence, let
\begin{eqnarray}
 \label{Eqn:QSTxProofCodeListSize}
 \alpha(m,s^{n}) \define \sum_{a \in \CalV^{k}}\mathds{1}_{\{ (v^{n}(a,m),s^{n}) \in  T_{\eta}^{n}(p_{VS})\}}
\end{eqnarray}
be the number of codewords in the bin indexed by the mesage that is jointly typical with the observed state sequence $s^{n} \in \CalS^{n}$. Let
\begin{eqnarray}
 \label{Eqn:QStxProofListDefn}
 \CalL(m,s^{n}) \define
 \begin{cases}
  \{ \substack{a : (v^{n}(a,m),s^{n}) \in T_{\eta}(p_{VS}) }\}&\mbox{if }{\alpha(m,s^{n}) \geq 1}\\
  \{0^{k}\}&\mbox{ otherwise, i.e. }{\alpha(m,s^{n}) = 0}.
  \end{cases}
\end{eqnarray}
be a list of candidate code words that is available to the encoder for the message, state sequence pair $(m,s^{n})$. Let $a^{*}_{\mbox{\tiny{$m,\!s^{n}$}}}$ be chosen from $\CalL(m,s^{n})$ and $v^{*}(m,s^{n})\define v^{n}(\astarmsn,m)$. A predefined `fusion map' $f:\CalS^{n} \times \CalV^{n} \rightarrow \CalX^{n}$ is used to map the pair $s^{n},v^{*}(m,s^{n})$ to an input sequence in $\CalX^{n}$ henceforth denoted $x^{n}(m,s^{n})$. On observing state sequence $s^{n}$ and message $m$, the encoder chooses state sequence $x^{n}(m,s^{n}) = (x(m,s^{n})_{t}:1\leq t \leq n)$, and we define $\rho_{m,s^{n}} \define \displaystyle \bigotimes_{t=1}^{n}\rho_{x(m,s^{n})_{t}s_{t}}$. 

\subsection{Decoding POVMs}
\label{SubSec:QSTxProofdecodingPOVM}
Consider a PMF $p_{SVX}=\ttp_{S}p_{VX|S}$ on $\CalS\times \CalV \times \CalX$. Let 
\begin{eqnarray}
 \label{Eqn:QSTxproofQuantumStates}
 \rho \define \sum_{x,s}p_{SX}(s,x)\rho_{xs}, \rho_{v}\define \sum_{x,s}p_{XS|V}(x,s|v)\rho_{xs}\!\!\!\begin{array}{c}\mbox{have spectral}\\\mbox{ decompositions }\end{array}\!\!\!\rho = \sum_{y \in \CalY}q(y)\ketbra{f_{y}}\mbox{ and }\rho_{v}= \sum_{y \in \CalY}r_{Y|V}(y|v)\ketbra{e_{y|v}}
 \nonumber
 \end{eqnarray}
respectively. Let
\begin{eqnarray}
 \label{Eqn:QSTxDecProofTypicalProjectors}
 \pi^{Y} \define \!\!\sum_{y^{n} \in \CalY^{n}}\!\!\bigotimes_{t=1}^{n}\ketbra{f_{y_{t}}}\mathds{1}_{\left\{ y^{n} \in T_{\eta}^{n}(q)\right\}}\mbox{ and }\pi_{v^{n}}\define \left\{ \!\!
 \begin{array}{lr}0&\mbox{if }v^{n} \notin T_{\eta}^{n}(p_{V})\\
 \!\!\displaystyle\sum_{y^{n} \in \CalY^{n}}\!\!\bigotimes_{t=1}^{n}\ketbra{e_{y_{t}|v_{t}}}\mathds{1}_{\{ (v^{n},y^{n})\in T_{\eta}^{n}(p_{V}r_{Y|V})\}}&\mbox{otherwise.}
 \end{array}\right.
\end{eqnarray}
be the unconditional and conditional typical projectors. For $(a,m) \in \CalV^{k}\times \CalV^{l}$, let
\begin{eqnarray}
 \label{Eqn:QSTxProofDecodingPOVM}
 \gamma_{a,m} \define \pi^{Y}\pi_{v^{n}(a,m)}\pi^{Y}\mbox{ and }\lambda_{m} \!\define\! \left(\sum_{\hata,\hatm\in \CalV^{k}\!\times\!\CalV^{l}}\!\!\!\!\!\!\gamma_{\hata,\hatm}\right)^{\!\!-\frac{1}{2}}\!\!\!\sum_{a\in \CalV^{k}}\gamma_{a,m} \left(\sum_{\hata,\hatm\in \CalV^{k}\!\times\!\CalV^{l}}\!\!\!\!\!\!\gamma_{\hata,\hatm}\right)^{\!\!-\frac{1}{2}}\!\!\!\!\mbox{ for }m \in [q^{l}]\mbox{ and }\lambda_{-1}=I_{\CalH_{Y}}^{\otimes n}-\sum_{m\in \CalM}\!\!\lambda_{m}
\end{eqnarray}
and $\{\lambda_{m}:m \in \CalM=[q^{l}],\lambda_{-1}\}$ be the decoding POVM.

\subsection{Error Probability}
\label{SubSec:QSTXProofErrorProb}
As is standard in information theory, we derive an upper bound on the error probability of the best code by averaging the error probability over an ensemble of codes. We begin by specifying the distribution of a random code in our ensemble. Note that a code is completely specified via (i) the generator matrix $g \in \CalV^{k \times n}$, the map $\iota : \CalV^{l}\rightarrow \CalV^{n}$ and the collection $(\astarmsn \in \CalV^{k}:(m,s^{n}) \in \CalM\times \CalS^{n})$. The generator matrix $G$, the map $\iota$ and the collection $(\Astarmsn \in \CalV^{k}:(m,s^{n}) \in \CalM\times \CalS^{n})$ of a random code are distributed with PMF
\begin{eqnarray}
 \label{Eqn:QSTxProofRandomCodeDistribution}
 P\left(\!\!\!\begin{array}{c}G=g,\iota(\tilde{m})=d^{n}(\tilde{m}):\tilde{m}\in\CalV^{l}, \\\astarmsn = a(m,s^{n}):(m,s^{n})\in \CalM\times\CalS^{n}\end{array}\!\!\!\right)=\frac{1}{q^{kn}}\left[\prod_{\tildem \in \CalV^{l}}\frac{1}{q^{n}} \right]\cdot \left[\prod_{m\in\CalV^{l}}\prod_{s^{n}\in\CalS^{n}}\frac{1}{|\CalL(m,s^{n})|}\mathds{1}_{\{a(m,s^{n}) \in \CalL(m,s^{n})\}}\right].
\end{eqnarray}
From \eqref{Eqn:QSTxProofRandomCodeDistribution}, it can be verified that the generator matrix $G$ and the range of $(\iota(m):m \in \CalV^{l})$ are mutually independent and uniformly distributed in the respective range spaces. Moreover, for $(m,s^{n}) \in \CalM \times \CalS^{n}$ and any $ a \in \CalL(m,s^{n})$, we note that
\begin{eqnarray}
 \label{Eqn:QSTxProofListErrorImport}
 P\left(\!\!\!\begin{array}{c}\astarmsn = a(m,s^{n})\mbox{ for}\\(m,s^{n})\in \CalM\times\CalS^{n} \end{array}\!\!\left|\!\!\begin{array}{c}
 G=g,\iota(\tilde{m})=d^{n}(\tilde{m})\\\mbox{for }\tilde{m}\in\CalV^{l} \end{array}\right.\!\!\!\right)=\frac{1}{|\CalL(m,s^{n})|}\mathds{1}_{\{a(m,s^{n}) \in \CalL(m,s^{n})\}},
\end{eqnarray}
a relation we shall opportunity to use in our analysis. For a specific code, the average error probability of the code is
\begin{eqnarray}
 \label{Eqn:QSTxProofErrProb1}
 \overline{\xi}(e,\lambda) \define \frac{1}{q^{l}}\sum_{m}\sum_{s^{n}}\ttp_{S}^{n}(s^{n})\tr\left\{ \left(\bI-\lambda_{m}\right)\rho_{m,s^{n}}\right\}\leq \ttt_{0}+ \frac{1}{q^{l}}\sum_{m}\sum_{s^{n}}\ttp_{S}^{n}(s^{n})\tr\left\{ \left(\bI-\lambda_{m}\right)\rho_{m,s^{n}}\right\}\mathds{1}_{\CalE_{L-\eta}}, \mbox{ where},\nonumber\\
 \label{Eqn:QSTxProofErrProb2}
\ttt_{0} \define \frac{1}{q^{l}}\sum_{m,s^{n}}\ttp_{S}^{n}(s^{n})\mathds{1}_{\CalE_{L-\eta}^{c}}, \CalE_{A} \define \left\{ \alpha(m,s^{n}) \geq 2^{nA}\right\}.\mbox{ Suppose }S = \gamma_{\astarmsn,m}, T = \!\!\!\sum_{a \neq \astarmsn}\!\!\!\gamma_{a,m}+\sum_{\hatm \neq m}\sum_{a}\gamma_{a,\hatm},\mbox{ then } \nonumber
\nonumber
\end{eqnarray}
$\lambda_{m} \geq (S+T)^{-\frac{1}{2}}S(S+T)^{-\frac{1}{2}}$ and the operator inequalities $0 \leq S \leq \bI$, $0 \leq T$ hold. In breaking down the error event, we have considered the event $\CalE^{L-\eta}$ and the choice of $L$ will be specified in due course. From the Hayashi Nagaoka inequality \cite{200307TIT_HayNag} \cite[Lemma 16.4.1]{BkWilde_2017}, we have
\begin{eqnarray}
\label{Eqn:QSTxApplyingHayashiNagaoka1}
I^{\otimes n} - \lambda_{m} \leq I^{\otimes n} - (S+T)^{-\frac{1}{2}}S(S+T)^{-\frac{1}{2}} \leq 2(I-S)+4(I-T)\mbox{ and hence }\overline{\xi}(e,\lambda) \leq \ttt_{0}+\ttt_{1}+\ttt_{2}\mbox{ where }
\nonumber\\
\label{Eqn:QSTxApplyingHayashiNagaoka2}
\ttt_{1}\define \frac{2}{q^{l}}\sum_{m}\sum_{s^{n}}\ttp_{S}^{n}(s^{n})\tr(\left[\bI-\gamma_{\astarmsn,m}\right]\rho_{m,s^{n}})\mathds{1}_{\CalE_{L-\eta}}, \ttt_{2}=\ttt_{21}+\ttt_{22}, \ttt_{21} \define \frac{4}{q^{l}}\sum_{m}\sum_{s^{n}}\sum_{\hata \neq \astarmsn}\!\!\!\!\ttp_{S}^{n}(s^{n})\tr(\gamma_{\hata,m}\rho_{m,s^{n}})\mathds{1}_{\CalE_{L-\eta}}\nonumber\\
\mbox{and }\ttt_{22}\define \frac{4}{q^{l}}\sum_{m}\sum_{s^{n}}\sum_{\hatm \neq m}\sum_{\hata}\ttp_{S}^{n}(s^{n})\tr(\gamma_{\hata,\hatm}\rho_{m,s^{n}})\mathds{1}_{\CalE_{L-\eta}}.
\nonumber
 \end{eqnarray}
Let $\ttT_{i}:0\leq i \leq 3$ denote abov terms corresponding to a random code whose distribution is specified in \eqref{Eqn:QSTxProofRandomCodeDistribution}. In the sequel, we derive upper bounds for each of the four terms $\ttt_{0},\ttt_{1},\ttt_{21},\ttt_{22}$ corresponding to the best code by evaluating upper bounds on $\Expectation\{\ttT_{i}\}:0 \leq i \leq 3$.
\med\textit{Bound on $\mathbb{E}\{\ttT_{0}\}$} : We note that $\ttt_{0}$ concerns only the event that the encoder cannot find atleast $2^{n(L-\eta)}$ codewords in the bin indexed by the message that is jointly typical with the observed state sequence. The analysis of this event involves only classical probabilities. Using a standard second moment method similar to that in \cite[Lemma 7 in Appendix B]{201710TIT_PadPra} or \cite[Lemma 8 in Appendix A]{201804TIT_PadPra}. Employing these techniques, the following lemma can be proved.
\begin{lemma}
 \label{Lem:}
 For every $\eta >0$, there exists $N_{\eta} \in \naturals$ such that for all $n \geq N_{\eta}$, we have
 \begin{eqnarray}
  \label{Eqn:QSTxProofBoundOnCovering}
  \Expectation\{\ttT_{0}\} \leq \exp\left\{-n\left( \frac{k\log q}{n}- \left[\log q -H(V|S)+L-\eta \right]  \right) \right\}.
 \end{eqnarray}
\end{lemma}
\med\textit{Bound on $\Expectation\{\ttT_{1}\}$} : For a specific code, we have $\ttt_{1}=\ttt_{11}+\ttt_{12}$, where 
\begin{eqnarray}
 \label{Eqn:QSTxProofAnalyzeT1_1}
 \label{Eqn:QSTxProofAnalyzeT1_2}
 \ttt_{11} &\define&  \frac{2}{q^{l}}\sum_{m}\sum_{s^{n}}\ttp_{S}^{n}(s^{n})\tr(\left[\bI-\gamma_{\astarmsn,m}\right]\rho_{m,s^{n}})\mathds{1}_{\left\{\!\!\!\begin{array}{c} \alpha(m,s^{n})\geq 2^{n(L-\eta)}, (s^{n},v^{*}(m,s^{n})) \in T_{\eta}^{n}(p_{SV})\\x^{n}(m,s^{n}) \notin T_{\eta}^{n}(p_{X|SV}|s^{n},v^{*}(m,s^{n}))\end{array}\!\!\! \right\}} \nonumber\\
  \label{Eqn:QSTxProofAnalyzeT1_4}
 \ttt_{12} &\define&  \frac{2}{q^{l}}\sum_{m}\sum_{s^{n}}\ttp_{S}^{n}(s^{n})\tr(\left[\bI-\gamma_{\astarmsn,m}\right]\rho_{m,s^{n}})\mathds{1}_{\left\{\!\!\!\begin{array}{c} \alpha(m,s^{n})\geq 2^{n(L-\eta)}, (s^{n},v^{*}(m,s^{n}),x^{n}(m,s^{n})) \in T_{\eta}^{n}(p_{SVX})\end{array}\!\!\! \right\}}.\nonumber
\end{eqnarray}
A bound on $\Expectation\{\ttT_{11}\}$ can be derived using standard bounds on typical sequences. Indeed, since $X^{n}(m,s^{n})$ is conditionally picked wrt $\prod p_{X|VS}(\cdot|v^{*}(m,s^{n}),s^{n})$, the probability that it is not conditionally typically falls exponentially. In the following, we indicate how we breakup $\ttt_{12}$ and indicate how each term in corresponding breakup can be bounded. 

\med\textit{Bound on $\Expectation\{\ttT_{12}\}$} : We have
\begin{eqnarray}
 \label{Eqn:QSTxProofAnalyzeT1_5}
  \ttt_{12}&\leq&  \frac{2}{q^{l}}\sum_{\substack{ m \in \CalV^{l}\\a \in \CalV^{k} }}\sum_{\substack{(s^{n},v^{n},x^{n})\\\in T_{\eta}^{n}(p_{SVX})}}\!\!\!\!\ttp_{S}^{n}(s^{n})\tr(\left[\bI-\pi^{Y}\pi_{v^{n}}\pi^{Y}\right]\rho_{x^{n},s^{n}})\mathds{1}_{\left\{\!\!\!\begin{array}{c} \astarmsn=a,v^{n}(a,m)=v^{n},x^{n}(m,s^{n})=x^{n}\end{array}\!\!\! \right\}}\nonumber\\
 \label{Eqn:QSTxProofAnalyzeT1_6}
  &\leq&  \frac{2}{q^{l}}\sum_{\substack{ m \in \CalV^{l}\\a \in \CalV^{k} }}\sum_{\substack{(s^{n},v^{n},x^{n})\\\in T_{\eta}^{n}(p_{SVX})}}\!\!\!\!\ttp_{S}^{n}(s^{n})\tr(\left[\bI-\pi^{Y}\pi_{v^{n}}\pi^{Y}\right]\rho_{x^{n},s^{n}})\mathds{1}_{\left\{\!\!\!\begin{array}{c} v^{n}(a,m)=v^{n},x^{n}(m,s^{n})=x^{n}\end{array}\!\!\! \right\}}=\ttt_{121}-\ttt_{122}\nonumber\\
 \label{Eqn:QSTxProofAnalyzeT1_7}
  \ttt_{121}&\define&\frac{2}{q^{l}}\sum_{\substack{ m \in \CalV^{l}\\a \in \CalV^{k} }}\sum_{\substack{(s^{n},v^{n},x^{n})\\\in T_{\eta}^{n}(p_{SVX})}}\!\!\!\!\ttp_{S}^{n}(s^{n})\mathds{1}_{\left\{\!\!\!\begin{array}{c} v^{n}(a,m)=v^{n},x^{n}(m,s^{n})=x^{n}\end{array}\!\!\! \right\}}
  \nonumber\\
  \label{Eqn:QSTxProofAnalyzeT1_8}
  \ttt_{122}&\define&\frac{2}{q^{l}}\sum_{\substack{ m \in \CalV^{l}\\a \in \CalV^{k} }}\sum_{\substack{(s^{n},v^{n},x^{n})\\\in T_{\eta}^{n}(p_{SVX})}}\!\!\!\!\ttp_{S}^{n}(s^{n})\tr(\left[\pi^{Y}\pi_{v^{n}}\pi^{Y}\right]\rho_{x^{n},s^{n}})\mathds{1}_{\left\{\!\!\!\begin{array}{c} v^{n}(a,m)=v^{n},x^{n}(m,s^{n})=x^{n}\end{array}\!\!\! \right\}}.\nonumber
\end{eqnarray}
A lower bound on $\Expectation\{\ttT_{122}\}$ can be derived using gentle operator lemma and pinching. Specifically, we have
\begin{eqnarray}
 \Expectation\{\ttT_{122}\} &=& \frac{2}{q^{l}}\sum_{\substack{ m \in \CalV^{l}\\a \in \CalV^{k} }}\sum_{\substack{(s^{n},v^{n},x^{n})\\\in T_{\eta}^{n}(p_{SVX})}}\!\!\!\!\ttp_{S}^{n}(s^{n})\tr(\left[\pi^{Y}\pi_{v^{n}}\pi^{Y}\right]\rho_{x^{n},s^{n}})P\left(\!\!\!\begin{array}{c} v^{n}(a,m)=v^{n},x^{n}(m,s^{n})=x^{n}\end{array}\!\!\! \right)\nonumber\\
 &=&\frac{2}{q^{l+n}}\sum_{\substack{ m \in \CalV^{l}\\a \in \CalV^{k} }}\sum_{\substack{(s^{n},v^{n},x^{n})\\\in T_{\eta}^{n}(p_{SVX})}}\!\!\!\!\ttp_{S}^{n}(s^{n})\tr(\pi_{v^{n}}\pi^{Y}\rho_{x^{n},s^{n}}\pi^{Y})p_{X|VS}(x^{n}|v^{n},s^{n})\nonumber\\
 \label{Eqn:QSTxProofAnalyzeT1_10}
 &\geq& \frac{2}{q^{l+n}}\sum_{\substack{ m \in \CalV^{l}\\a \in \CalV^{k} }}\sum_{\substack{(s^{n},v^{n},x^{n})\\\in T_{\eta}^{n}(p_{SVX})}}\!\!\!\!\ttp_{S}^{n}(s^{n})\left[\tr(\pi_{v^{n}}\rho_{x^{n},s^{n}})-\norm{\pi^{Y}\rho_{x^{n},s^{n}}\pi^{Y} - \rho_{x^{n},s^{n}}}_{1}\right]p_{X|VS}(x^{n}|v^{n},s^{n}).
\end{eqnarray}
The first term on the RHS of \eqref{Eqn:QSTxProofAnalyzeT1_10} can be lower bounded by the pinching lemma \cite[Property 15.2.7]{BkWilde_2017} and the second term can be upper bounded via gentle operator lemma \cite[Lemma 9.4.2]{BkWilde_2017}. We now proceed to $\Expectation\{\ttT_{11}\}$

\med\textit{Bound on $\ttt_{11}$} : 
\begin{eqnarray}
 \label{Eqn:QSTxProofAnalyzeT1_3}
 \ttt_{11}\leq \frac{2}{q^{l}}\sum_{m}\sum_{s^{n}}\ttp_{S}^{n}(s^{n})\mathds{1}_{\left\{\!\!\!\begin{array}{c} (s^{n},v^{*}(m,s^{n})) \in T_{\eta}^{n}(p_{SV}),x^{n}(m,s^{n}) \notin T_{\eta}^{n}(p_{X|SV}|s^{n},v^{*}(m,s^{n}))\end{array}\!\!\! \right\}} \mbox{ and }\nonumber\\
\end{eqnarray}

In the sequel, we compute upper bound on $\Expectation\{{\ttT}_{11}\}$ and lower bound on $\Expectation\{\overline{\ttT}_{12}\}$. We begin with the latter. We have
\begin{eqnarray}
 \label{Eqn:QSTxProofAnalyzeT1_2}
  \Expectation\{\overline{\ttt}_{12}\} &=&\frac{2}{q^{l}}\sum_{m}\sum_{\substack{(s^{n},v^{n},x^{n})\\\in T_{\eta}(p_{SVX})}}\ttp_{S}^{n}(s^{n})\tr(\pi^{Y}\pi_{v^{n}}\pi^{Y}\rho_{x^{n},s^{n}})P(X^{n}(m,s^{n})=x^{n},V^{*}(m,s^{n})=v^{n})
  \nonumber\\
  &=&\frac{2}{q^{l}}\sum_{m}\sum_{\substack{(s^{n},v^{n},x^{n})\\\in T_{\eta}(p_{SVX})}}\ttp_{S}^{n}(s^{n})\tr(\pi^{Y}\pi_{v^{n}}\pi^{Y}\rho_{x^{n},s^{n}})p_{X|VS}(x^{n}|v^{n},s^{n})P(V^{*}(m,s^{n})=v^{n})\nonumber\\
  &\geq& \frac{2}{q^{l}}\sum_{m}\sum_{\substack{(s^{n},v^{n},x^{n})\\\in T_{\eta}(p_{SVX})}}\ttp_{S}^{n}(s^{n})\tr(\pi^{Y}\pi_{v^{n}}\pi^{Y}\rho_{x^{n},s^{n}})p_{X|VS}(x^{n}|v^{n},s^{n})P(V^{*}(m,s^{n})=v^{n})
\end{eqnarray}
\med\textit{Bound on $\Expectation\{\ttT_{21}\}$} : Before we study $\Expectation\{\ttT_{21}\}$, we begin by noting that
\begin{eqnarray}
 \label{Eqn:QSTxProofAnalysisT21_1}
 \ttt_{21} = \frac{4}{q^{l}}\sum_{m}\sum_{s^{n}}\sum_{\hata \neq \astarmsn}\!\!\!\!\ttp_{S}^{n}(s^{n})\tr(\gamma_{\hata,m}\rho_{m,s^{n}})\mathds{1}_{\CalE_{L-\eta}} = \frac{4}{q^{l}}\sum_{m}\sum_{s^{n}}\sum_{\hata \neq \astarmsn}\!\!\!\!\ttp_{S}^{n}(s^{n})\tr(\pi^{Y}\pi_{v^{n}(\hata,m)}\pi^{Y}\rho_{m,s^{n}})\mathds{1}_{\CalE_{L-\eta}}\\
 \label{Eqn:QSTxProofAnalysisT21_2}
 = \frac{4}{q^{l}}\sum_{m}\sum_{s^{n}}\sum_{\substack{a,\hata \in\CalV^{k}\\a \neq \hata}}\sum_{v^{n},\hatv^{n},x^{n}}\!\!\!\ttp_{S}^{n}(s^{n})\tr(\pi^{Y}\pi_{\hatv^{n}}\pi^{Y}\rho_{x^{n},s^{n}})\mathds{1}_{\left\{ \!\!\!
 \begin{array}{c} 
 \alpha(m,s^{n}) \geq 2^{n(L-\eta)}, \astarmsn=a, v^{n}(a,m)=v^{n}\\v^{n}(\hata,m)=\hatv^{n},x^{n}(m,s^{n})=x^{n}
 \end{array} \!\!\!\right\}}
 \\
 \label{Eqn:QSTxProofAnalysisT21_3}
 = \frac{4}{q^{l}}\sum_{m}\sum_{s^{n}}\sum_{\substack{a,\hata \in\CalV^{k}\\a \neq \hata}}\sum_{\substack{v^{n}\\\hatv^{n},x^{n}}}\!\!\!\ttp_{S}^{n}(s^{n})\tr(\pi^{Y}\pi_{\hatv^{n}}\pi^{Y}\rho_{x^{n},s^{n}})\mathds{1}_{\left\{ \!\!\!
 \begin{array}{c} 
 \alpha(m,s^{n}) \geq 2^{n(L-\eta)}, \astarmsn=a, v^{n}(a,m)=v^{n},\hatv^{n} \in T_{\eta}(p_{V})\\v^{n}(\hata,m)=\hatv^{n},x^{n}(m,s^{n})=x^{n},(v^{n},s^{n}) \in T_{\eta}^{n}(p_{VS})
 \end{array} \!\!\!\right\}}
\end{eqnarray}
where (i) $\rho_{m,s^{n}} = \displaystyle \bigotimes_{t=1}^{n}\rho_{x(m,s^{n})_{t}s_{t}}$ is as defined earlier, (ii) \eqref{Eqn:QSTxProofAnalysisT21_2} follows by summing over all possible choices for the corresponding codewords, (iii) \eqref{Eqn:QSTxProofAnalysisT21_3} holds for $L\geq \eta$ since the encoding rule guarantees $\astarmsn \in \CalL(m,s^{n})$ and the latter set defined in \eqref{Eqn:QStxProofListDefn} contains indices corresponding to codewords that are jointly typical with the observed state sequence whenever $\alpha(m,s^{n})\geq 1$ and (iv) \eqref{Eqn:QSTxProofAnalysisT21_3} is true since, as defined in \eqref{Eqn:QSTxProofDecodingPOVM}, $\pi_{\hatv^{n}}$ is the zero projector if $\hatv^{n} \notin T_{\eta}^{n}(p_{V})$. This implies
\begin{eqnarray}
  \label{Eqn:QSTxProofAnalysisT21_4}
\Expectation\{\ttT_{21}\} = \frac{4}{q^{l}}\sum_{m}\!\!\!\!\sum_{\substack{(v^{n},s^{n})\\\in T_{\eta}^{n}(p_{VS})}}\!\sum_{\substack{a,\hata \in\CalV^{k}\\a \neq \hata}}\sum_{\substack{\hatv^{n} \in T_{\eta}^{n}(p_{V})\\x^{n}\in \CalX^{n}}}\!\!\!\!\!\!\ttp_{S}^{n}(s^{n})\tr(\pi^{Y}\pi_{\hatv^{n}}\pi^{Y}\rho_{x^{n},s^{n}})P\left( \!\!\!\!
 \begin{array}{c} 
 \alpha(m,s^{n}) \geq 2^{n(L-\eta)}, \astarmsn=a, v^{n}(a,m)=v^{n}\\v^{n}(\hata,m)=\hatv^{n},x^{n}(m,s^{n})=x^{n},
 \end{array} \!\!\!\right)\!.
  \end{eqnarray}
For $(v^{n},s^{n}) \in T_{\eta}^{n}(p_{VS})$ and $\hata \neq a$, we have
\begin{eqnarray}
 \label{Eqn:QSTxProofAnalysisT21_5}
P\left( \!\!\!
 \begin{array}{c} 
 \alpha(m,s^{n}) \geq 2^{n(L-\eta)}, \astarmsn=a, V^{n}(a,m)=v^{n}\\V^{n}(\hata,m)=\hatv^{n},X^{n}(m,s^{n})=x^{n},
 \end{array} \!\!\!\right) =
 P\left( \!\!\!
 \begin{array}{c} 
 V^{n}(a,m)=v^{n}\\V^{n}(\hata,m)=\hatv^{n}
 \end{array} \!\!\!\right)P\left( \!\!\!
 \begin{array}{c} 
 \alpha(m,s^{n}) \geq\\ 2^{n(L-\eta)} \end{array}\left|\begin{array}{c}V^{n}(a,m)=v^{n}\\V^{n}(\hata,m)=\hatv^{n}
 \end{array}\right. \!\!\!\right)
 \\
 \label{Eqn:QSTxProofAnalysisT21_6}
 \times P\left(\!\!\!\! \begin{array}{c} \astarmsn=a\end{array}\!\!\! \left| \!\!\!\begin{array}{c} \alpha(m,s^{n}) \geq 2^{n(L-\eta)}\\V^{n}(a,m)=v^{n},V^{n}(\hata,m)=\hatv^{n} \end{array} \right.\!\!\!\!\right)P\left( \!\!\!\!\begin{array}{c}X^{n}(m,s^{n})=x^{n} \end{array} \!\!\left| \!\!\begin{array}{c} \astarmsn=a, \alpha(m,s^{n}) \geq 2^{n(L-\eta)}\\V^{n}(a,m)=v^{n},V^{n}(\hata,m)=\hatv^{n} \end{array} \right. \!\!\!\!\right)
 \nonumber\\
 \label{Eqn:QSTxProofAnalysisT21_7}
 \leq \frac{1}{q^{2n}}\cdot 1 \cdot \frac{1}{\alpha(m,s^{n})}p_{X|SV}(x^{n}|s^{n},v^{n})\mathds{1}_{\{\alpha(m,s^{n}) \geq 2^{n(L-\eta)} \}} \leq \frac{1}{q^{2n}} \cdot \frac{p_{X|SV}(x^{n}|s^{n},v^{n})}{2^{n(L-\eta)}} \leq  \frac{2^{n(H(V|S)+3\eta)}p_{XV|S}(x^{n},v^{n}|s^{n})}{q^{2n}2^{n(L-\eta)}},
 \end{eqnarray}
where the first inequality in \eqref{Eqn:QSTxProofAnalysisT21_7} follows the fact that codewords of a UCC are pairwise independent and uniformly distributed \cite{BkNIT_PraPadShi} and the second inequality in \eqref{Eqn:QSTxProofAnalysisT21_7} follows from standard bounds on conditional probability of jointly typical sequences. Substituting the bound in the RHS of \eqref{Eqn:QSTxProofAnalysisT21_7} in \eqref{Eqn:QSTxProofAnalysisT21_4}, we have
\begin{eqnarray}
\label{Eqn:QSTxProofAnalysisT21_8}
\Expectation\{\ttT_{21}\} \!\!\!&\leq&\!\!\! \frac{4\cdot2^{n(H(V|S)+3\eta)} }{q^{l+2n}\cdot 2^{n(L-\eta)}}\sum_{m}\!\!\!\!\sum_{\substack{(v^{n},s^{n})\\\in T_{\eta}^{n}(p_{VS})}}\!\sum_{\substack{a,\hata \in\CalV^{k}\\a \neq \hata}}\sum_{\substack{\hatv^{n} \in T_{\eta}^{n}(p_{V})\\x^{n}\in \CalX^{n}}}\tr(\pi^{Y}\pi_{\hatv^{n}}\pi^{Y} p_{XVS}(x^{n},v^{n},s^{n}) \rho_{x^{n},s^{n}})
 \nonumber\\
 \label{Eqn:QSTxProofAnalysisT21_9}
 \lefteqn{\!\!\!\!\!\!\leq \frac{4\cdot2^{n(H(V|S)+3\eta)} }{q^{l+2n}\cdot 2^{n(L-\eta)}}\sum_{m}\sum_{\substack{a,\hata \in\CalV^{k}\\a \neq \hata}}\sum_{\substack{\hatv^{n} \in\\ T_{\eta}^{n}(p_{V})}}\!\!\!\!\tr(\pi^{Y}\pi_{\hatv^{n}}\pi^{Y} \rho^{\otimes n}) = \frac{4\cdot2^{n(H(V|S)+3\eta)} }{q^{l+2n}\cdot 2^{n(L-\eta)}}\sum_{m}\sum_{\substack{a,\hata \in\CalV^{k}\\a \neq \hata}}\sum_{\substack{\hatv^{n} \in\\ T_{\eta}^{n}(p_{V})}}\!\!\!\!\tr(\pi_{\hatv^{n}}\pi^{Y} \rho^{\otimes n}\pi^{Y}) }
 \\
 \label{Eqn:QSTxProofAnalysisT21_10}
 \!\!\!\!\!\!&\leq&\!\! \frac{4\cdot2^{n(H(V|S)-H(Y)+6\eta)} }{q^{l+2n}\cdot 2^{n(L-\eta)}}\sum_{m}\sum_{\substack{a,\hata \in\CalV^{k}\\a \neq \hata}}\sum_{\substack{\hatv^{n} \in\\ T_{\eta}^{n}(p_{V})}}\!\!\!\tr(\pi_{\hatv^{n}}\pi^{Y})\leq \frac{4\cdot2^{n(H(V|S)-H(Y)+6\eta)} }{q^{l+2n}\cdot 2^{n(L-\eta)}}\sum_{m}\sum_{\substack{a,\hata \in\CalV^{k}\\a \neq \hata}}\sum_{\substack{\hatv^{n} \in \\T_{\eta}^{n}(p_{V})}}\!\!\!\tr(\pi_{\hatv^{n}})\\
 \label{Eqn:QSTxProofAnalysisT21_11}
 \!\!\!\!\!\!&\leq&\!\! \frac{4\cdot2^{n(H(V|S)-H(Y)+H(V,Y)+12\eta)} }{q^{l+2n}\cdot 2^{n(L-\eta)}}\sum_{m}\sum_{\substack{a,\hata \in\CalV^{k}\\a \neq \hata}}\!\!\!1 \leq \frac{4\cdot2^{n(H(V|S)-H(Y)+H(V,Y)+12\eta)} }{q^{2n-2k}\cdot 2^{n(L-\eta)}}\\
 \label{Eqn:QSTxProofAnalysisT21_12}
 \!\!\!\!\!\!&\leq&\!\! \exp\left\{ -n\left( L -\left[\frac{k\log q}{n}-\log q+H(V|S)\right]+\log q -H(V|Y)-\frac{k \log}{n}  \right)\right\}
 \end{eqnarray}
where \eqref{Eqn:QSTxProofAnalysisT21_9} follows from the operator inequality
\begin{eqnarray}
 \label{Eqn:QSTxProofAnalysisT21_13}
 \sum_{\substack{(v^{n},s^{n})\\\in T_{\eta}^{n}(p_{VS})}}\sum_{{x^{n}\in \CalX^{n}}} p_{XVS}(x^{n},v^{n},s^{n}) \rho_{x^{n},s^{n}} \leq \sum_{\substack{x^{n},s^{n},v^{n}\\ \in \CalX^{n}\times \CalS^{n}\times \CalV^{n}}}p_{XVS}(x^{n},v^{n},s^{n}) \rho_{x^{n},s^{n}} = \rho^{\otimes n}
 \nonumber
\end{eqnarray}
which follows from the positivity of the density operators, \eqref{Eqn:QSTxProofAnalysisT21_10} follows from the operator inequalities $\pi^{Y}\rho^{\otimes n}\pi^{Y} \leq 2^{-n[H(Y)-3\eta]}\pi^{Y}$ \cite[Property 15.1.3]{BkWilde_2017} and $\pi^{Y} \leq \bI$, \eqref{Eqn:QSTxProofAnalysisT21_11} follows from $\tr(\pi_{\hatv^{n}}) \leq 2^{n[H(Y|V)+3\eta]}$ for $\hatv^{n} \in T_{\eta}^{n}(p_{V})$\cite[Property 15.1.2]{BkWilde_2017} and $|T_{\eta}^{n}(p_{V})|\leq 2^{n[H(V)+3\eta]}$ and the last bound \eqref{Eqn:QSTxProofAnalysisT21_12} follows by collating all exponents.

\med\textit{Bound on $\Expectation\{\ttT_{22}\}$} : Our analysis of $\Expectation\{\ttT_{22}\}$ is very similar to $\Expectation\{\ttT_{21}\}$. The only difference between these analyses stems from the fact that the legitimate codeword $V^{n}(\astarmsn,m)$ and an incorrect codeword in the same bin $V^{n}(\hata,m)$ are not statistically independent, however the legitimate codeword $V^{n}(\astarmsn,m)$ and any codeword in a different bin $V^{n}(\hata,\hatm)$ for $\hatm \neq m$ are statistically independent. This suggests that we can derive the bounds without having to condition on the realization of $\astarmsn$ as in \eqref{Eqn:QSTxProofAnalysisT21_5} - \eqref{Eqn:QSTxProofAnalysisT21_7}. Except for this minor difference, the rest of the analysis provided below is identical. We have
\begin{eqnarray}
 \label{Eqn:QSTxProofAnalysisT22_1}
 \ttt_{22} = \frac{4}{q^{l}}\sum_{\substack{m,\hatm\\m\neq \hatm}}\sum_{s^{n}}\sum_{\hata}\ttp_{S}^{n}(s^{n})\tr(\gamma_{\hata,\hatm}\rho_{m,s^{n}})\mathds{1}_{\CalE_{L-\eta}} = \frac{4}{q^{l}}\sum_{\substack{m,\hatm\\m\neq \hatm}}\sum_{s^{n}}\sum_{\hata}\ttp_{S}^{n}(s^{n})\tr(\pi^{Y}\pi_{v^{n}(\hata,\hatm)}\pi^{Y}\rho_{m,s^{n}})\mathds{1}_{\CalE_{L-\eta}}\\
 \label{Eqn:QSTxProofAnalysisT22_2}
 = \frac{4}{q^{l}}\sum_{\substack{m,\hatm\\m\neq \hatm}}\sum_{s^{n}}\sum_{\hata}\sum_{v^{n},\hatv^{n},x^{n}}\!\!\!\ttp_{S}^{n}(s^{n})\tr(\pi^{Y}\pi_{\hatv^{n}}\pi^{Y}\rho_{x^{n},s^{n}})\mathds{1}_{\left\{ \!\!\!
 \begin{array}{c} 
 \alpha(m,s^{n}) \geq 2^{n(L-\eta)}, V^{*}(m,s^{n})=v^{n}\\v^{n}(\hata,\hatm)=\hatv^{n},x^{n}(m,s^{n})=x^{n}
 \end{array} \!\!\!\right\}}
 \\
 \label{Eqn:QSTxProofAnalysisT22_3}
 = \frac{4}{q^{l}}\sum_{\substack{m,\hatm\\m\neq \hatm}}\sum_{s^{n}}\sum_{\hata}\sum_{\substack{v^{n}\\\hatv^{n},x^{n}}}\!\!\!\ttp_{S}^{n}(s^{n})\tr(\pi^{Y}\pi_{\hatv^{n}}\pi^{Y}\rho_{x^{n},s^{n}})\mathds{1}_{\left\{ \!\!\!
 \begin{array}{c} 
 \alpha(m,s^{n}) \geq 2^{n(L-\eta)}, V^{*}(m,s^{n})=v^{n},\hatv^{n} \in T_{\eta}(p_{V})\\v^{n}(\hata,\hatm)=\hatv^{n},x^{n}(m,s^{n})=x^{n},(v^{n},s^{n}) \in T_{\eta}^{n}(p_{VS})
 \end{array} \!\!\!\right\}}
\end{eqnarray}
where as earlier, (i) $\rho_{m,s^{n}} = \displaystyle \bigotimes_{t=1}^{n}\rho_{x(m,s^{n})_{t}s_{t}}$, (ii) \eqref{Eqn:QSTxProofAnalysisT22_2} follows by summing over all possible choices for the corresponding codewords, (iii) \eqref{Eqn:QSTxProofAnalysisT22_3} holds for $L\geq \eta$ since the encoding rule guarantees $\astarmsn \in \CalL(m,s^{n})$ and the latter set defined in \eqref{Eqn:QStxProofListDefn} contains indices corresponding to codewords that are jointly typical with the observed state sequence whenever $\alpha(m,s^{n})\geq 1$ and (iv) \eqref{Eqn:QSTxProofAnalysisT22_3} is true since, as defined in \eqref{Eqn:QSTxProofDecodingPOVM}, $\pi_{\hatv^{n}}$ is the zero projector if $\hatv^{n} \notin T_{\eta}^{n}(p_{V})$. This implies
\begin{eqnarray}
  \label{Eqn:QSTxProofAnalysisT22_4}
\Expectation\{\ttT_{22}\} = \frac{4}{q^{l}}\sum_{\substack{m,\hatm\\m\neq \hatm}}\sum_{\substack{(v^{n},s^{n})\\\in T_{\eta}^{n}(p_{VS})}}\!\sum_{\hata }\sum_{\substack{\hatv^{n} \in T_{\eta}^{n}(p_{V})\\x^{n}\in \CalX^{n}}}\!\!\!\!\!\!\ttp_{S}^{n}(s^{n})\tr(\pi^{Y}\pi_{\hatv^{n}}\pi^{Y}\rho_{x^{n},s^{n}})P\left( \!\!\!\!
 \begin{array}{c} 
 \alpha(m,s^{n}) \geq 2^{n(L-\eta)}, V^{*}(m,s^{n})=v^{n}\\V^{n}(\hata,m)=\hatv^{n},x^{n}(m,s^{n})=x^{n} \end{array} \!\!\!\right)\!.
  \end{eqnarray}
For $(v^{n},s^{n}) \in T_{\eta}^{n}(p_{VS})$ and $\hatm \neq m$, we have
\begin{eqnarray}
 \label{Eqn:QSTxProofAnalysisT22_5}
P\left( \!\!\!\!
 \begin{array}{c} 
 \alpha(m,s^{n}) \geq 2^{n(L-\eta)}, V^{*}(m,s^{n})=v^{n}\\V^{n}(\hata,\hatm)=\hatv^{n},X^{n}(m,s^{n})=x^{n} \end{array} \!\!\!\right)&=&
\sum_{a \in \CalV^{k}}P\left( \!\!\!
 \begin{array}{c} 
 \alpha(m,s^{n}) \geq 2^{n(L-\eta)}, \astarmsn=a, V^{n}(a,m)=v^{n}\\V^{n}(\hata,\hatm)=\hatv^{n},X^{n}(m,s^{n})=x^{n},
 \end{array} \!\!\!\right)
\nonumber \\
 \label{Eqn:QSTxProofAnalysisT22_7}
 &\leq&  \frac{2^{n(H(V|S)+3\eta)}p_{XV|S}(x^{n},v^{n}|s^{n})}{q^{-k+2n}2^{n(L-\eta)}},
 \end{eqnarray}
where \eqref{Eqn:QSTxProofAnalysisT22_7} follows from the same set of arguments that got us from \eqref{Eqn:QSTxProofAnalysisT21_5} to \eqref{Eqn:QSTxProofAnalysisT21_7}. Substituting the above upper bound in \eqref{Eqn:QSTxProofAnalysisT22_4}, we have
\begin{eqnarray}
\label{Eqn:QSTxProofAnalysisT22_8}
\!\!\!\!\!\!\!\!\!\!\!\!\!\!\Expectation\{\ttT_{22}\} \!\!\!&\leq&\!\!\! \frac{4\cdot 2^{n(H(V|S)+3\eta)} }{q^{-k+l+2n}\cdot 2^{n(L-\eta)}}\sum_{\substack{m,\hatm\\m\neq \hatm}}\sum_{\substack{(v^{n},s^{n})\\\in T_{\eta}^{n}(p_{VS})}}\!\sum_{\hata }\sum_{\substack{\hatv^{n} \in T_{\eta}^{n}(p_{V})\\x^{n}\in \CalX^{n}}}\!\!\!\!\!\!\tr(\pi^{Y}\pi_{\hatv^{n}}\pi^{Y}p_{XVS}(x^{n},v^{n},s^{n})\rho_{x^{n},s^{n}})
 \nonumber\\
 \label{Eqn:QSTxProofAnalysisT22_9}
 \lefteqn{\!\!\!\!\!\!\!\!\!\!\!\!\!\!\!\!\!\!\!\!\!\leq \frac{4\cdot 2^{n(H(V|S)+3\eta)} }{q^{-k+l+2n}\cdot 2^{n(L-\eta)}}\sum_{\substack{m,\hatm\\m\neq \hatm}}\sum_{\hata }\sum_{\substack{\hatv^{n} \in \\T_{\eta}^{n}(p_{V})}}\!\!\!\tr(\pi^{Y}\pi_{\hatv^{n}}\pi^{Y} \rho^{\otimes n}) = \frac{4\cdot 2^{n(H(V|S)+3\eta)} }{q^{-k+l+2n}\cdot 2^{n(L-\eta)}}\sum_{\substack{m,\hatm\\m\neq \hatm}}\sum_{\hata }\sum_{\substack{\hatv^{n} \in \\T_{\eta}^{n}(p_{V})}}\!\!\!\tr(\pi_{\hatv^{n}}\pi^{Y} \rho^{\otimes n}\pi^{Y}) }
 \\
 \label{Eqn:QSTxProofAnalysisT22_10}
 \!\!\!\!\!\!\!\!\!&\leq&\!\! \frac{4\cdot 2^{n(H(V|S)+3\eta)} }{q^{-k+l+2n}\cdot 2^{n(L-\eta)}}\sum_{\substack{m,\hatm\\m\neq \hatm}}\sum_{\hata }\sum_{\substack{\hatv^{n} \in \\T_{\eta}^{n}(p_{V})}}\!\!\!\tr(\pi_{\hatv^{n}}\pi^{Y})\leq \frac{4\cdot 2^{n(H(V|S)+3\eta)} }{q^{-k+l+2n}\cdot 2^{n(L-\eta)}}\sum_{\substack{m,\hatm\\m\neq \hatm}}\sum_{\hata }\sum_{\substack{\hatv^{n} \in \\T_{\eta}^{n}(p_{V})}}\!\!\!\tr(\pi_{\hatv^{n}})\\
 \label{Eqn:QSTxProofAnalysisT22_11}
 \!\!\!\!\!\!\!\!\!\!&\leq&\!\! \frac{4\cdot 2^{n(H(V|S)+3\eta)} }{q^{-k+l+2n}\cdot 2^{n(L-\eta)}}\sum_{\substack{m,\hatm\\m\neq \hatm}}\sum_{\hata }1 \leq \frac{4\cdot2^{n(H(V|S)-H(Y)+H(V,Y)+12\eta)} }{q^{2n-2k-l}\cdot 2^{n(L-\eta)}}\\
 \label{Eqn:QSTxProofAnalysisT22_12}
 \!\!\!\!\!\!\!\!\!\!&\leq&\!\! \exp\left\{ -n\left( L -\left[\frac{k\log q}{n}-\log q+H(V|S)\right]+\log q -H(V|Y)-\frac{k \log}{n} -\frac{l \log}{n}  \right)\right\}
 \end{eqnarray}
where \eqref{Eqn:QSTxProofAnalysisT22_9} follows from the operator inequality
\begin{eqnarray}
 \label{Eqn:QSTxProofAnalysisT21_13}
 \sum_{\substack{(v^{n},s^{n})\\\in T_{\eta}^{n}(p_{VS})}}\sum_{{x^{n}\in \CalX^{n}}} p_{XVS}(x^{n},v^{n},s^{n}) \rho_{x^{n},s^{n}} \leq \sum_{\substack{x^{n},s^{n},v^{n}\\ \in \CalX^{n}\times \CalS^{n}\times \CalV^{n}}}p_{XVS}(x^{n},v^{n},s^{n}) \rho_{x^{n},s^{n}} = \rho^{\otimes n}
 \nonumber
\end{eqnarray}
which follows from the positivity of the density operators, \eqref{Eqn:QSTxProofAnalysisT22_10} follows from the operator inequalities $\pi^{Y}\rho^{\otimes n}\pi^{Y} \leq 2^{-n[H(Y)-3\eta]}\pi^{Y}$ \cite[Property 15.1.3]{BkWilde_2017} and $\pi^{Y} \leq \bI$, \eqref{Eqn:QSTxProofAnalysisT22_11} follows from $\tr(\pi_{\hatv^{n}}) \leq 2^{n[H(Y|V)+3\eta]}$ for $\hatv^{n} \in T_{\eta}^{n}(p_{V})$\cite[Property 15.1.2]{BkWilde_2017} and $|T_{\eta}^{n}(p_{V})|\leq 2^{n[H(V)+3\eta]}$ and the last bound \eqref{Eqn:QSTxProofAnalysisT22_12} follows by collating all exponents.
\subsection{Collating Bounds and Characterization of a Single-letter achievable Rate Region}
\label{SubSec:QSTxProofCollatingBounds}
Through the above analysis, we have proved that for every choice of a finite field $\CalV=\fieldq$ and a PMF $p_{SVX}$ on $\CalS\times \CalV \times \CalX$, there exists a code of block length $n$ specified through an encoder $e$, a decoding POVM $\lambda$ consisting of $q^{l}$ codewords with error probability
\begin{eqnarray}
 \label{Eqn:QSTxProofFinalBounds1}
 \overline{\xi}(e,\lambda) &\leq& \exp\left\{-n\left( \frac{k\log q}{n}- \left[\log q -H(V|S)+L-\eta \right]  \right) \right\}\nonumber\\
  \label{Eqn:QSTxProofFinalBounds2}
 &&+\exp\left\{ -n\left( L -\left[\frac{k\log q}{n}-\log q+H(V|S)\right]+\log q -H(V|Y)-\frac{k \log}{n}  \right)\right\}\nonumber\\
  \label{Eqn:QSTxProofFinalBounds3}
 &&+\exp\left\{ -n\left( L -\left[\frac{k\log q}{n}-\log q+H(V|S)\right]+\log q -H(V|Y)-\frac{k \log}{n} -\frac{l \log}{n}  \right)\right\}
 \nonumber
\end{eqnarray}
if $L \geq \eta > 0$. By choosing $L \define \frac{k\log q}{n}- \log q +H(V|S)-4\eta \mbox{ we can guarantee } \overline{\xi}(e,\lambda) \leq 3\exp\left\{ -n\eta\right\}\mbox{ if }$
\begin{eqnarray}
   \label{Eqn:QSTxProofFinalBounds5}
 \frac{k\log q}{n}- \log q+H(V|S) > 5\eta >0\mbox{ and }\frac{(k+l)\log q}{n} < \log q -H(V|Y)
\end{eqnarray}
where are information quantities are computed with respect to the state defined in \eqref{Eqn:QSTxThmStatementQuantumState}. We therefore choose
\begin{eqnarray}
 \label{Eqn:QSTxProofFinalBounds6}
 \frac{k\log q}{n}= \log q-H(V|S)+5\eta \mbox{ and }\frac{l\log q}{n} > H(V|S) -H(V|Y)-5\eta = I(V;Y)-I(V;S)-5\eta
\end{eqnarray}
and guarantee $\overline{\xi}(e,\lambda) \leq 3\exp\left\{ -n\eta\right\}$. This completes the proof.
\end{proof}

\appendices
\section{Characterization of the Quantum States in evaluation of Information Quantities for Ex.~\ref{Ex:BDDQMSTx}}
\label{AppSec:BDDQMSTxExampleQuantumStates}
Consider Ex.~\ref{Ex:BDDQMSTx} for $\theta \in (0,\frac{\pi}{2})$. In this appendix, we provide characterization of the quantum state in \eqref{Eqn:QuantStateInIIDCodeRateRegion} for the choice $\CalU_{1}=\CalU_{2}=\{0,1\}$, $p_{U_{j}|S_{j}}(1|0)=p_{U_{j}|S_{j}}(0|1)=\tau = 1-p_{U_{j}|S_{j}}(0|0)=1-p_{U_{j}|S_{j}}(1|1)$ and $X_{j}=U_{j}\oplus S_{j}$ for $j \in [2]$, where $\oplus$ denotes addition mod$-2$. The characterizations below enable us compute the information quantities and thereby quantify the upper bound on the sum rate achievable via IID random codes. The latter is stated in our discussion prior to Sec.~\ref{SubSec:RoleOfAlgebraicCodes}. For the choice of parameters stated earlier, the quantum state in \eqref{Eqn:QuantStateInIIDCodeRateRegion} is
\begin{eqnarray}
 \sigma^{YS_{1}S_{2}X_{1}X_{2}U_{1}U_{2}} \!=\! \sum_{s_{1},s_{2}}\! \frac{\tau(1\!-\!\tau)}{4}\!\left[ \mathds{1}_{\left\{\substack{s_{1}\oplus s_{2}\\=0}\right\}}\ketbra{1}+\mathds{1}_{\left\{\substack{s_{1}\oplus s_{2}\\=1}\right\}}\ketbra{v_{\theta}} \right]\otimes\ketbra{s_{1}~s_{2}}\otimes \left[\!\!\begin{array}{c}\ketbra{0~1~s_{1}~1\oplus s_{2}}+\\\ketbra{1~0~1\oplus s_{1}~s_{2}} \end{array}\!\!\right]
 \nonumber\\
 + \sum_{s_{1},s_{2}}\left[ \mathds{1}_{\left\{\substack{s_{1}\oplus s_{2}\\=0}\right\}}\ketbra{0}+\mathds{1}_{\left\{\substack{s_{1}\oplus s_{2}\\=1}\right\}}\ketbra{v_{\theta}^{\perp}} \right]\otimes \ketbra{s_{1}~s_{2}}\otimes\left[\!\!\begin{array}{c}\frac{(1-\tau)^{2}}{4}\ketbra{0~0~s_{1}~s_{2}}+\\\frac{\tau^{2}}{4}\ketbra{1~1~1\oplus s_{1}~1\oplus s_{2}} \end{array}\!\!\right].
 \nonumber
\end{eqnarray}
Partial tracing over the appropriate component systems, we have
\begin{eqnarray}
 \sigma^{S_{1}S_{2}U_{1}U_{2}} \!\!&\!\!=\!\!&\!\! \sum_{s_{1},s_{2}}\frac{\tau(1-\tau)}{4}\left( \ketbra{s_{1}~s_{2}~1\oplus s_{1}~s_{2}} + \ketbra{s_{1}~s_{2}~s_{1}~1\oplus s_{2}}\right)\nonumber\\
 \!\!&\!\!\!\!&\!\!+  \sum_{s_{1},s_{2}}\frac{\tau^{2}}{4} \ketbra{s_{1}~s_{2}~1\otimes s_{1}~1\oplus s_{2}} + \sum_{s_{1},s_{2}}\frac{(1-\tau)^{2}}{4} \ketbra{s_{1}~s_{2}~ s_{1}~ s_{2}}\nonumber
 \mbox{ implying}\\
 \sigma^{S_{j}U_{j}} \!\!&\!\!=\!\!&\!\!\sum_{s_{j}}\frac{\tau(1-\tau)+\tau^{2}}{2}\ketbra{s_{j}~1\oplus s_{j}}+\sigma^{S_{j}U_{j}} =\sum_{s_{j}}\frac{\tau(1-\tau)+(1-\tau)^{2}}{2}\ketbra{s_{j}~ s_{j}}\nonumber\\
 \!\!&\!\!=\!\!&\!\!\frac{\tau}{2}\ketbra{0~1}+\frac{\tau}{2}\ketbra{1~0}+\frac{1-\tau}{2}\ketbra{0~0}+\frac{1-\tau}{2}\ketbra{1~1}\mbox{ for $j \in [2]$ and }\nonumber\\
 \sigma^{YU_{1}U_{2}} \!\!&\!\!=\!\!&\!\! \sum_{s_{1},s_{2}}\! \frac{\tau(1\!-\!\tau)}{4}\!\left[ \mathds{1}_{\left\{\substack{s_{1}\oplus s_{2}\\=0}\right\}}\ketbra{1}+\mathds{1}_{\left\{\substack{s_{1}\oplus s_{2}\\=1}\right\}}\ketbra{v_{\theta}} \right]\otimes \left[\!\!\begin{array}{c}\ketbra{s_{1}~1\oplus s_{2}}+\ketbra{1\oplus s_{1}~s_{2}} \end{array}\!\!\right]
 \nonumber\\
 \!\!&\!\!\!\!&\!\!+ \sum_{s_{1},s_{2}}\left[ \mathds{1}_{\left\{\substack{s_{1}\oplus s_{2}\\=0}\right\}}\ketbra{0}+\mathds{1}_{\left\{\substack{s_{1}\oplus s_{2}\\=1}\right\}}\ketbra{v_{\theta}^{\perp}} \right]\otimes \left[\!\!\begin{array}{c}\frac{(1-\tau)^{2}}{4}\ketbra{s_{1}~s_{2}}+\frac{\tau^{2}}{4}\ketbra{1\oplus s_{1}~1\oplus s_{2}} \end{array}\!\!\right]\nonumber\end{eqnarray}\begin{eqnarray}
 \!\!&\!\!=\!\!&\!\!\frac{2\tau(1-\tau)}{4}\ketbra{1}\otimes \left( \ketbra{0~1}+\ketbra{1~0}\right)+\frac{2\tau(1-\tau)}{4}\ketbra{v_{\theta}}\otimes \left( \ketbra{0~0}+\ketbra{1~1}\right)\nonumber\\
 \!\!&\!\!\!\!&\!\!+\left[ \frac{(1-\tau)^{2}+\tau^{2}}{4} \right]\left[ \ketbra{0}\otimes \left(\ketbra{0~0}+\ketbra{1~1} \right) +\ketbra{v_{\theta}^{\perp}}\left( \ketbra{0~1}+\ketbra{1~0}\right) \right]\mbox{ implying}
 \nonumber\\
\lefteqn{\!\!\!\!\!\!\!\!\!\!\!\!\!\!\!\!\!\!\!\!\!\!\!=\frac{\left( \epsilon \ketbra{1}+(1-\epsilon)\ketbra{v_{\theta}^{\perp}}\right)}{4}\otimes\left(\ketbra{0~1}+\ketbra{1~0}\right)+\frac{\left( \epsilon \ketbra{v_{\theta}}+(1-\epsilon)\ketbra{0}\right)}{4}\otimes\left(\ketbra{0~0}+\ketbra{1~1}\right)\mbox{ implying}}
 \nonumber\\
 \sigma^{Y}\!\!&\!\!=\!\!&\!\! \frac{\epsilon}{2} \ketbra{1}+\frac{(1-\epsilon)}{2}\ketbra{v_{\theta}^{\perp}} +\frac{\epsilon}{2} \ketbra{v_{\theta}}+\frac{(1-\epsilon)}{2}\ketbra{0},~~\sigma^{U_{1}U_{2}}= \frac{1}{4}\sum_{u_{1},u_{2}}\ketbra{u_{1}~u_{2}}\nonumber
\end{eqnarray}
where $\epsilon = 2\tau(1-\tau)$.

\section{Proof of Prop.~\ref{Prop:Step1ProoT1} : Bound on $T_{2}$}
\label{AppSec:BoundOnT2}
We begin by defining events
\begin{eqnarray}
 \lefteqn{\CalF_{1} \define \left\{\substack{V_{j}^{n}(a_{j},m_{j})=v_{j}^{n}:j \in [2],\ulineS^{n}=\ulines^{n}\\W^{n}(\hata,\hum)=\hatw^{n},W^{n}(\apl,\ulinem)=w^{n} }\right\}, \CalF_{2}\define \left\{ \substack{A_{j}(m_{j},s_{j}^{n})\\=a_{j}:j \in [2]}\right\}\cap\CalE}
 \nonumber\\
 &&\!\!\!\!\!\!\!\!\!\!\!\!\!\!\!\CalF_{3}\! \define \!\left\{ \!\substack{X_{j}^{n}(m_{j},s_{j}^{n})\\=x_{j}^{n}: j \in [2]} \!\right\}\!, \beta \define \left\{\!\substack{(v_{j}^{n},s_{j}^{n})\in T_{\eta}(p_{V_{j}S_{j}}\!) 
 \\ w^{n}=v_{1}^{n}\oplus_{q}v_{2}^{n},} \!\right\}\!, \omega \define \left\{ \substack{w^{n} \in T_{\eta}(p_{W})\\\hatw^{n} \in T_{\eta}(p_{W})}\right\}.
\nonumber
 \end{eqnarray}
From the definition of $\apl$ and the distribution of the random code, we have
\begin{eqnarray}
\label{Eqn:T2Events-1}
 \lefteqn{P(\CalF_{1}\cap\CalF_{2}\cap\CalF_{3})\mathds{1}_{\beta}\mathds{1}_{\omega} 
 \leq P(\CalF_{1}) P(\CalF_{3}|\CalF_{1}\cap\CalF_{2})\mathds{1}_{\beta}\mathds{1}_{\omega}} \\
 \label{Eqn:T2Events-2}
 &&\!\!\!\!\!\!\!\!\!\!\!\leq  \frac{1}{q^{3n}}\ttp_{\ulineS}(\ulines^{n})\prod_{j=1}^{2}p_{X_{j}|V_{j}S_{j}}^{n}(x_{j}^{n}|v_{j}^{n},s_{j}^{n})\mathds{1}_{\beta}\mathds{1}_{\omega}
 \\
 \label{Eqn:T2Events-3}
 &&\!\!\!\!\!\!\!\!\!\!\!\leq \frac{2^{n(H(V_{1}|S_{1})+2\eta)}}{q^{3n}2^{-n(H(V_{2}|S_{2}))}}\ttp_{\ulineS}(\ulines^{n})\!\prod_{j=1}^{2}p_{X_{j}V_{j}|S_{j}}^{n}\!(x_{j}^{n},v_{j}^{n}|s_{j}^{n})\mathds{1}_{\beta}\mathds{1}_{\omega}\\
 \label{Eqn:T2Events-4}
 &&\!\!\!\!\!\!\!\!\!\!\!= \Theta p_{\ulineS\ulineV\ulineX}^{n}(\ulines^{n},\ulinev^{n},\ulinex^{n})\mathds{1}_{\beta}\mathds{1}_{\omega}\mbox{ where } \Theta \define \frac{2^{n(H(V_{1}|S_{1})+2\eta)}}{q^{3n}2^{-n(H(V_{2}|S_{2}))}}
\end{eqnarray}
where \eqref{Eqn:T2Events-2} folows from the property of a uniformly distributed UCC proven in \cite[Lemma 9 in Appendix E]{201710TIT_PadPra}, \eqref{Eqn:T2Events-3} follows from the presence of $\mathds{1}_{\beta}$ in the factors. The term corresponding to $T_{2}$ from \eqref{Eqn:ProofStep1ErrAnalysis-11}, \eqref{Eqn:ProofStep1ErrAnalysis-1} in $\Expectation\{\xi(\ulinee,\lambda) \}$ is
\begin{eqnarray}
\label{Eqn:BoundOnT2_1}
 \lefteqn{\sum_{\ulines^{n}}\!\ttp_{\ulineS}(\ulines^{n}\!)\Expectation\{ T_{2}\} \!=\!\!\!\!\!\!\!\!\!\!\!\!\!\!\!\!\!\sum_{\substack{a_{1},a_{2},\ulinev^{n},\ulines^{n},\ulinex^{n},w^{n}\\(\hata,\hum) \neq (\apl,\ulinem),\hatw^{n} }}\!\!\!\!\!\!\!\!\!\!\!\!\!\!\!\!\! \tr(\pi^{Y}\!\pi_{\hatw^{n}}\pi^{Y}\!\rho_{\ulinex^{n},\ulines^{n}})P\left(\!\bigcap_{k=1}^{3}\CalF_{k}\!\!\right)\mathds{1}_{\beta}\mathds{1}_{\omega}} \nonumber\\
 \label{Eqn:BoundOnT2_2}
 &&\!\!\!\!\!\!\!\!\!\!\!\leq\Theta\!\!\!\!\!\!\!\sum_{\substack{a_{1},a_{2},w^{n},\hatw^{n}\\(\hata,\hum) \neq (\apl,\ulinem) }} \!\!\!\!\!
 \!\!\!\!\!\!\!\mathds{1}_{\omega} \tr(\!\!\pi^{Y}\!\pi_{\hatw^{n}}\pi^{Y}\!\!\!\!\!\sum_{\substack{\ulinev^{n},\ulines^{n},\ulinex^{n} }}\!\!\!\!\!\!p_{\ulineS\ulineV\ulineX}^{n}(\ulines^{n},\ulinev^{n},\ulinex^{n})\rho_{\ulinex^{n},\ulines^{n}}\mathds{1}_{\beta}\!\!)\\
 \label{Eqn:BoundOnT2_3}
 &&\!\!\!\!\!\!\!\!\!\!\!\leq\Theta\!\!\!\!\!\!\!\sum_{\substack{a_{1},a_{2},w^{n},\hatw^{n}\\(\hata,\hum) \neq (\apl,\ulinem) }} \!\!\!\!\!
 \!\!\!\!\!\!\!\mathds{1}_{\omega} \tr(\pi^{Y}\!\pi_{\hatw^{n}}\pi^{Y}p_{W}^{n}(w^{n})\rho_{w^{n}})\\
 \label{Eqn:BoundOnT2_4}
 &&\!\!\!\!\!\!\!\!\!\!\!\leq \Theta\!\!\!\!\!\!\!\sum_{\substack{a_{1},a_{2},\hatw^{n}\\(\hata,\hum) \neq (\apl,\ulinem) }} \!\!\!\!\!
 \!\!\!\!\!\!\!\mathds{1}_{\omega} \tr(\pi^{Y}\!\pi_{\hatw^{n}}\pi^{Y}\rho^{\otimes n})=\Theta\!\!\!\!\!\!\!\!\!\sum_{\substack{a_{1},a_{2},\hatw^{n}\\(\hata,\hum) \neq (\apl,\ulinem) }} \!\!\!\!\!
 \!\!\!\!\!\!\!\mathds{1}_{\omega} \tr(\pi_{\hatw^{n}}\pi^{Y}\rho^{\otimes n}\pi^{Y}) \nonumber \\
 \label{Eqn:BoundOnT2_5}
 &&\!\!\!\!\!\!\!\!\!\!\!\leq \frac{\Theta}{2^{nH(Y)_{\sigma}}}\!\!\!\!\!\!\!\sum_{\substack{a_{1},a_{2},\hatw^{n}\\(\hata,\hum) \neq (\apl,\ulinem) }} \!\!\!\!\!
 \!\!\!\!\!\!\!\mathds{1}_{\omega} \tr(\pi_{\hatw^{n}}\pi^{Y})=\frac{\Theta q^{k_{1}+2k_{2}}2^{n(H(Y,W)_{\sigma}+6\eta)}}{2^{nH(Y)_{\sigma}}q^{-l_{1}-l_{2}}}\\
 \label{Eqn:BoundOnT2_6}
 &&\!\!\!\!\!\!\!\!\!\!\!\leq \exp\left\{-n\left({3\log q ~-~ {H(W|Y)_{\sigma}+\sum_{i=1}^{2}H(U_{i}|S_{i})_{\sigma}} - \frac{k_{1}+2k_{2}+l_{1}+l_{2}}{n}} \log q -8\eta \right) \right\}\nonumber
\end{eqnarray}
where \eqref{Eqn:BoundOnT2_2} follows by substituting the upper bound \eqref{Eqn:T2Events-4}, \eqref{Eqn:BoundOnT2_3} follows from averaging the density operators and the fact that density operators are positive, \eqref{Eqn:BoundOnT2_4} follows again by averaging and cyclicity of the trace, \eqref{Eqn:BoundOnT2_5} follows from the operator inequality $\pi^{Y}\rho^{\otimes n}\pi^{Y} \leq 2^{-n(H(Y)_{\sigma}-2\eta)}\pi^{Y}$ and the fact that for typical $\hatw^{n}$, we have $\tr(\pi_{\hatw^{n}}\pi^{Y}) \leq \tr(\pi_{\hatw^{n}})\leq 2^{nH(Y|W)+2n\eta}$ and the last inequality follows by substituting the value of $\Theta$ from \eqref{Eqn:T2Events-4}. We obtained the above bound on $\frac{k_{1}+2k_{2}+l_{1}+l_{2}}{n} \log q$ since we have assumed $k_{2} \geq k_{1}$. In general, we obtain the bound 
\begin{eqnarray}
 \sum_{\ulines^{n}}\ttp_{\ulineS}(\ulines^{n}\!)\Expectation\{ T_{2}\} \leq \exp\left\{-n\left({3\log q ~-~ {H(W|Y)_{\sigma}+\sum_{i=1}^{2}H(U_{i}|S_{i})_{\sigma}} - \frac{k_{1}+k_{2}+\max\{ k_{1},k_{2}\}+l_{1}+l_{2}}{n}} \log q -8\eta \right) \right\}\nonumber
\end{eqnarray}

\bibliographystyle{IEEEtran}
{
\bibliography{/home/arunpr/work/qislBib/qisl}

\begin{thebibliography}{10}
\providecommand{\url}[1]{#1}
\csname url@rmstyle\endcsname
\providecommand{\newblock}{\relax}
\providecommand{\bibinfo}[2]{#2}
\providecommand\BIBentrySTDinterwordspacing{\spaceskip=0pt\relax}
\providecommand\BIBentryALTinterwordstretchfactor{4}
\providecommand\BIBentryALTinterwordspacing{\spaceskip=\fontdimen2\font plus
\BIBentryALTinterwordstretchfactor\fontdimen3\font minus
  \fontdimen4\font\relax}
\providecommand\BIBforeignlanguage[2]{{%
\expandafter\ifx\csname l@#1\endcsname\relax
\typeout{** WARNING: IEEEtran.bst: No hyphenation pattern has been}%
\typeout{** loaded for the language `#1'. Using the pattern for}%
\typeout{** the default language instead.}%
\else
\language=\csname l@#1\endcsname
\fi
#2}}

\bibitem{1980MMPCIT_GelPin}
S.~I. Gel'fand and M.~S. Pinsker, ``Capacity of a broadcast channel with one
  deterministic component,'' \emph{Probl. Pered. Inform.}, vol.~16, no.~1, pp.
  24--34, Jan.-Mar. 1980, ; translated in \emph{Probl. Inform. Transm.}, vol.
  16, no. 1, pp. 17-25, Jan.-Mar. 1980.

\bibitem{201710TIT_PadPra}
A.~Padakandla and S.~S. Pradhan, ``An {A}chievable {R}ate {R}egion {B}ased on
  {C}oset {C}odes for {M}ultiple {A}ccess {C}hannel {W}ith {S}tates,''
  \emph{IEEE Transactions on Information Theory}, vol.~63, no.~10, pp.
  6393--6415, Oct 2017.

\bibitem{BkNIT_PraPadShi}
\BIBentryALTinterwordspacing
S.~S. Pradhan, A.~Padakandla, and F.~Shirani, ``An algebraic and probabilistic
  framework for network information theory,'' \emph{Foundations and Trends® in
  Communications and Information Theory}, vol.~18, no.~2, pp. 173--379, 2020.
  [Online]. Available: \url{http://dx.doi.org/10.1561/0100000083}
\BIBentrySTDinterwordspacing

\bibitem{202107ISIT_AnwPadPra3CQIC}
T.~A. Atif, A.~Padakandla, and S.~S. Pradhan, ``{A}chievable rate-region for
  3—{U}ser {C}lassical-{Q}uantum {I}nterference {C}hannel using {S}tructured
  {C}odes,'' in \emph{2021 IEEE International Symposium on Information Theory
  (ISIT)}, 2021, pp. 760--765.

\bibitem{202107ISIT_AnwPadPraComMAC}
------, ``Computing {S}um of {S}ources over a {C}lassical-{Q}uantum
  {M}{A}{C},'' in \emph{2021 IEEE International Symposium on Information Theory
  (ISIT)}, 2021, pp. 414--419.

\bibitem{197903TIT_KorMar}
J.~Korner and K.~Marton, ``How to encode the modulo-two sum of binary sources
  (corresp.),'' \emph{IEEE Transactions on Information Theory}, vol.~25, no.~2,
  pp. 219--221, 1979.

\bibitem{200710TIT_NazGas}
B.~Nazer and M.~Gastpar, ``Computation over multiple-access channels,''
  \emph{IEEE Transactions on Information Theory}, vol.~53, no.~10, pp.
  3498--3516, 2007.

\bibitem{201804TIT_PadPra}
A.~{Padakandla} and S.~S. {Pradhan}, ``{A}chievable {R}ate {R}egion for {T}hree
  {U}ser {D}iscrete {B}roadcast {C}hannel {B}ased on {C}oset {C}odes,''
  \emph{IEEE Transactions on Information Theory}, vol.~64, no.~4, pp.
  2267--2297, April 2018.

\bibitem{201603TIT_PadSahPra}
A.~Padakandla, A.~G. Sahebi, and S.~S. Pradhan, ``An {A}chievable {R}ate
  {R}egion for the {T}hree-{U}ser {I}nterference {C}hannel {B}ased on {C}oset
  {C}odes,'' \emph{IEEE Transactions on Information Theory}, vol.~62, no.~3,
  pp. 1250--1279, March 2016.

\bibitem{202202arXiv_Pad3CQBC}
\BIBentryALTinterwordspacing
A.~Padakandla, ``An achievable rate region for $3-$user classical-quantum
  broadcast channels,'' 2022. [Online]. Available:
  \url{https://arxiv.org/abs/2203.00110}
\BIBentrySTDinterwordspacing

\bibitem{201706ISIT_HeiShiPra}
M.~Heidari, F.~Shirani, and S.~S. Pradhan, ``A new achievable rate region for
  multiple-access channel with states,'' in \emph{2017 IEEE International
  Symposium on Information Theory (ISIT)}, 2017, pp. 36--40.

\bibitem{201910TIT_HeiShiPra}
------, ``Quasi structured codes for multi-terminal communications,''
  \emph{IEEE Transactions on Information Theory}, vol.~65, no.~10, pp.
  6263--6289, 2019.

\bibitem{195804IBMJRD_Sha}
C.~E. Shannon, ``Channels with side information at the transmitter,'' \emph{IBM
  Journal of Research and Development}, vol.~2, no.~4, pp. 289--293, 1958.

\bibitem{197905TIT_Mar}
K.~{Marton}, ``A coding theorem for the discrete memoryless broadcast
  channel,'' \emph{IEEE Transactions on Information Theory}, vol.~25, no.~3,
  pp. 306--311, May 1979.

\bibitem{201604JPA_BocCaiNot}
\BIBentryALTinterwordspacing
H.~Boche, N.~Cai, and J.~Nötzel, ``The classical-quantum channel with random
  state parameters known to the sender,'' \emph{Journal of Physics A:
  Mathematical and Theoretical}, vol.~49, no.~19, p. 195302, apr 2016.
  [Online]. Available: \url{https://doi.org/10.1088/1751-8113/49/19/195302}
\BIBentrySTDinterwordspacing

\bibitem{201310JPA_Not}
J.~N\"otzel, ``Hypothesis testing on invariant subspaces of the symmetric
  group, part i - quantum sanov's theorem and arbitrarily varying sources,''
  \emph{Journal of Physics A: Mathematical and Theoretical}, vol.~47, 10 2013.

\bibitem{201206TIT_FawHaySavSenWil}
O.~{Fawzi}, P.~{Hayden}, I.~{Savov}, P.~{Sen}, and M.~M. {Wilde}, ``Classical
  communication over a quantum interference channel,'' \emph{IEEE Transactions
  on Information Theory}, vol.~58, no.~6, pp. 3670--3691, 2012.

\bibitem{201512TIT_SavWil}
I.~Savov and M.~M. Wilde, ``Classical codes for quantum broadcast channels,''
  \emph{IEEE Transactions on Information Theory}, vol.~61, no.~12, pp.
  7017--7028, 2015.

\bibitem{200906TIT_PhiZam}
T.~Philosof and R.~Zamir, ``On the loss of single-letter characterization: The
  dirty multiple access channel,'' \emph{IEEE Trans. on Info. Th.}, vol.~55,
  pp. 2442--2454, June 2009.

\bibitem{197401TIT_Wyn}
A.~Wyner, ``Recent results in the shannon theory,'' \emph{Information Theory,
  IEEE Transactions on}, vol.~20, no.~1, pp. 2 -- 10, Jan 1974.

\bibitem{200307TIT_HayNag}
M.~Hayashi and H.~Nagaoka, ``General formulas for capacity of classical-quantum
  channels,'' \emph{IEEE Transactions on Information Theory}, vol.~49, no.~7,
  pp. 1753--1768, 2003.

\bibitem{BkWilde_2017}
M.~M. Wilde, \emph{Quantum Information Theory}, 2nd~ed.\hskip 1em plus 0.5em
  minus 0.4em\relax Cambridge University Press, 2017.

\end{thebibliography}
\end{document}